\documentclass[sigconf]{acmart} %

\renewcommand\footnotetextcopyrightpermission[1]{}

\usepackage{xurl}
\usepackage{xspace} 
\usepackage{cleveref}
\usepackage{amsmath}
\usepackage{tabularx}
\usepackage{multirow}
\usepackage{hhline}
\usepackage{stackengine}
\usepackage{boldline}
\usepackage{balance}

\newcommand{\labelthis}[1]{ \stepcounter{equation}\tag{\theequation}\label{#1}}

\newcommand{\edit}[1]{#1}

\newcommand{\true}{\ensuremath{\mathsf{T}}\xspace} 
\newcommand{\false}{\ensuremath{\mathsf{F}}\xspace} 
\newcommand{\dist}{\ensuremath{\mathcal{D}}\xspace}
\newcommand{\e}{\ensuremath{\mathop{\mathbb{E}}}\xspace} %
\newcommand{\prob}{\ensuremath{\mathsf{Pr}}\xspace} %

\newcommand{\den}{\ensuremath{\mathsf{p}}\xspace} %

\DeclareMathOperator*{\argmin}{arg\,min}
\DeclareMathOperator{\var}{Var}
\DeclareMathOperator{\cov}{Cov}

\newcommand{\coco}{\textsc{CoCo}}

\newcommand{\sys}{\ensuremath{s}\xspace} %
\newcommand{\sysmod}{\ensuremath{\hat{s}}\xspace} %
\newcommand{\sysreal}{\ensuremath{s^*}\xspace} %

\newcommand{\states}{\ensuremath{X}\xspace} %
\newcommand{\statevec}{\ensuremath{\boldsymbol{x}}\xspace} %
\newcommand{\statevecs}{\ensuremath{\boldsymbol{X}}\xspace} %

\newcommand{\obss}{\ensuremath{Y}\xspace} %
\newcommand{\obs}{\ensuremath{y}\xspace} %
\newcommand{\obsvec}{\ensuremath{\boldsymbol{y}}\xspace} %
\newcommand{\obsvecs}{\ensuremath{\boldsymbol{Y}}\xspace} %

\newcommand{\ctrl}{\ensuremath{h}\xspace} %

\newcommand{\acts}{\ensuremath{U}\xspace} %

\newcommand{\dynm}{\ensuremath{f_d}\xspace} %
\newcommand{\dynms}{\ensuremath{F_d}\xspace} %

\newcommand{\measm}{\ensuremath{f_m}\xspace}
\newcommand{\measms}{\ensuremath{F_m}\xspace} %

\newcommand{\prop}{\ensuremath{\varphi}\xspace} %
\newcommand{\assn}{\ensuremath{A}\xspace} %
\newcommand{\verif}{\ensuremath{V}\xspace} %
\newcommand{\propform}{\ensuremath{\psi}\xspace} %

\newcommand{\mon}{\ensuremath{M}\xspace} %

\newcommand{\safe}{\ensuremath{\Phi}\xspace} %

\newcommand{\compassn}{\ensuremath{\assn_\propform}\xspace} %
\newcommand{\compfun}{\ensuremath{C}\xspace} %

\newcommand{\compmon}{\ensuremath{\mon_\compfun}\xspace} %

\newcommand{\ecerr}{\ensuremath{ECE}\xspace}%
\newcommand{\mcerr}{\ensuremath{MCE}\xspace}
\newcommand{\conserr}{\ensuremath{CCE}\xspace}%

\newcommand{\indep}{\ensuremath{\bot}\xspace}%

\newcommand{\eat}[1]{}

\makeatletter
\newtheorem*{rep@theorem}{\rep@title}
\newcommand{\newreptheorem}[2]{%
\newenvironment{rep#1}[1]{%
 \def\rep@title{#2 \ref{##1}}%
 \begin{rep@theorem}}%
 {\end{rep@theorem}}}
\makeatother

\newreptheorem{theorem}{Theorem}
\newreptheorem{lemma}{Lemma}

\def\smallint{\begingroup\textstyle \int\endgroup}
\newcommand{\psize}{\small}
\newcommand{\cex}[1]{E_{#1}}

\newcommand{\eecerr}{\ensuremath{\widehat{\ecerr}}\xspace}%
\newcommand{\emcerr}{\ensuremath{\widehat{MCE}}\xspace}
\newcommand{\econserr}{\ensuremath{\widehat{CCE}}\xspace}%
\newcommand{\ebrier}{\ensuremath{\widehat{Brier}}\xspace} %

\newcommand{\emon}{\ensuremath{\widehat{\mon}}\xspace} %

\newcommand{\eassn}{\ensuremath{\widehat{\assn}}\xspace} %

\newcommand{\esafe}{\ensuremath{\widehat{\safe}}\xspace} %

\newcommand{\eauc}{\ensuremath{\widehat{AuC}}\xspace}

\newtheorem{theorem}{Theorem}

\newtheorem{corollary}[theorem]{Corollary}

\newtheorem{definition}{Definition}

\newtheorem{lemma}{Lemma}
\newtheorem*{lemma*}{Lemma}
\newtheorem*{problem*}{Problem}

\acmYear{2022}
\acmDOI{}
\acmISBN{}
\acmConference[ICCPS '22]{ICCPS '22: ACM/IEEE International Conference on Cyber-physical Systems}{May 04--06, 2022}{Milan, Italy}
\acmBooktitle{ICCPS '22: ACM/IEEE International Conference on Cyber-physical Systems,
  May 04--06, 2022, Milan, Italy}

\setcopyright{none}

\begin{document}

\title{Confidence Composition for \\
Monitors of Verification Assumptions}

\author{Ivan Ruchkin}
\affiliation{%
  \institution{University of Pennsylvania}
  \city{Philadelphia}
  \state{Pennsylvania}
  \country{}
  \postcode{19104}}

\author{Matthew Cleaveland}
\affiliation{%
  \institution{University of Pennsylvania}
  \city{Philadelphia}
  \state{Pennsylvania}
\country{}
  \postcode{19104}}
  
\author{Radoslav Ivanov}
\affiliation{%
  \institution{Rensselaer Polytechnic Inst.}
  \city{Troy}
  \state{New York}
\country{}
  \postcode{19104}}
  
\author{Pengyuan Lu}
\affiliation{%
  \institution{University of Pennsylvania}
  \city{Philadelphia}
  \state{Pennsylvania}
\country{}
  \postcode{19104}}

\author{Taylor Carpenter}
\affiliation{%
  \institution{University of Pennsylvania}
  \city{Philadelphia}
  \state{Pennsylvania}
\country{}
  \postcode{19104}}
  
\author{Oleg Sokolsky}
\affiliation{%
  \institution{University of Pennsylvania}
  \city{Philadelphia}
  \state{Pennsylvania}
\country{}
  \postcode{19104}}

\author{Insup Lee}
\affiliation{%
  \institution{University of Pennsylvania}
  \city{Philadelphia}
  \state{Pennsylvania}
\country{}
  \postcode{19104}}

\renewcommand{\shortauthors}{Ruchkin et al.}

\begin{abstract}
\looseness=-1
Closed-loop verification of cyber-physical systems with neural network controllers offers strong safety guarantees under certain assumptions. It is, however, difficult to determine whether these guarantees apply at run time because verification assumptions may be violated. To predict safety violations in a verified system, we propose a three-step \edit{confidence composition (\coco)} framework for monitoring verification assumptions. First, we represent the sufficient condition for verified safety with a propositional logical formula over assumptions. Second, we build calibrated confidence monitors that evaluate the probability that each assumption holds. Third, we obtain the confidence in the verification guarantees by composing the assumption monitors using a composition function suitable for the logical formula. \edit{Our \coco} framework provides theoretical bounds on the calibration and conservatism of compositional monitors. \edit{Two case studies show that compositional monitors are calibrated better than} their constituents and successfully predict safety violations.
\end{abstract}


\maketitle

\renewcommand{\state}{\ensuremath{x}\xspace} %

\section{Introduction}

\looseness=-1
Autonomous cyber-physical systems, such as self-driving cars and service robots, are increasingly deployed in complex and safety-critical environments in our society~\cite{watanabe_runtime_2018,schumann_model-based_2019,fisac_probabilistically_2018}. Recently, the breakthrough capabilities in handling such environments came from the use of learning components, which may behave unpredictably. To consistently rely on such capabilities, one needs to ensure that the system would not to endanger the lives and property around it, or at least that an early enough warning is given to avert the disaster. %

\looseness=-1
When assuring a complex cyber-physical system, one can obtain strong safety guarantees from closed-loop reachability verification, recently extended to explicitly check neural network (NN) controllers~\cite{ivanov19,tran_nnv_2020}. To provide its guarantees, the verification relies on assumptions about system's dynamics, perception, and environment. Should the system find itself in circumstances not matching these assumptions, the verification guarantees are void --- and remarkably difficult to re-obtain at run time due to limited scalability.  

\looseness=-1
On another front, many run-time monitoring techniques were developed to detect anomalies, such as model inconsistencies and  out-of-distribution samples  ~\cite{boursinos_assurance_2021,carpenter_modelguard_2021,park_pac_2021}. These tools can provide valuable situational insights, but their outputs often lack a direct connection to the verification guarantees or system-level safety. For example, it is not clear to which extent an out-of-distribution image of a stop sign invalidates a collision-safety guarantee for an autonomous car. %

Thus, it is both challenging and important to quantify and monitor the trust in design-time verification guarantees at run time. In particular, it is vital to know when the guarantees no longer apply, so as to switch to a backup controller, execute a recovery maneuver, or ask for human assistance. The monitoring of verification guarantees has the potential to predict otherwise unforeseen failures in situations for which the system was not trained or designed.%

\looseness=-1
To monitor verification guarantees, we propose quantifying the \emph{confidence} in the \emph{assumptions} of verification. By confidence we mean an estimate of the probability that the assumption holds. Although an assumption may not be directly observable, its monitor would over time  accumulate confidence, which, if properly calibrated, \edit{would be} close to the true chance of satisfying the assumption given the observations. If all the assumptions are satisfied, our verification retroactively guarantees safety. \edit{Such  assumptions can be monitored} with off-the-shelf techniques~\cite{carpenter_modelguard_2021,dellaert99,wald04,poor13,scharf91}, and \edit{their confidences} would be combined into \edit{a} single confidence in the guarantees. In a safety-critical system, this confidence should not be over-estimated.

This paper introduces the \edit{\coco} \emph{framework for \underline{co}mposing \underline{co}nfi-dences} from monitors of verification assumptions,  consisting of three steps: (i) \emph{verify} the system under explicit assumptions, such that a propositional formula over these assumptions entails the system's safety, (ii) build a well-calibrated \emph{confidence monitor} for each assumption, (iii) use a \emph{composition function} informed by the formula from the first step to combine the monitor outputs into a composed confidence. This confidence quantifies the chance that the verification guarantees apply at that moment. 

We develop the theoretical conditions under which the composed confidence is calibrated and conservative, up to a bounded error, with respect to the true probability of safety. These conditions are that (a) the system model under verification can explain most safe behaviors, (b) a violation of assumptions would likely lead to a failure, and (c) the composition function is calibrated/conservative with respect to the assumptions. We also prove calibration error bounds for two composition functions --- product and weighted average --- and a conservatism bound for the product. 

\looseness=-1
We evaluate \edit{\coco} on \edit{two systems} with NN controllers: a mountain car and an underwater vehicle. Experiments show that our compositional monitors are useful for safety prediction, outperform the individual monitors, and can be tuned \edit{for conservatism}. Our data-driven composition functions improve the performance further if provided \edit{the data relating the monitors and the assumptions}.  

To summarize, this paper makes \edit{four} contributions:
\begin{itemize}
\looseness-1
    \item The \edit{\coco}~framework for \edit{composing} confidence monitors of verification assumptions \edit{with} five composition functions.
    \looseness-1
    \item Sufficient conditions for bounded calibration of composite confidence to the safety chance, \edit{with} the expectations from models, assumptions, monitors, and composition functions.  
    \item Upper bounds on the expected calibration error of two composition functions, and the conservatism error of one. 
    \item \edit{Two case studies that demonstrate} the utility of the framework and the trade-offs between composition functions. 
\end{itemize}

The rest of the paper proceeds as follows. The necessary background on verification and monitoring is given in the next section. \Cref{sec:related} surveys the existing research. \Cref{sec:fwk} presents the key contributions: the framework, the end-to-end calibration conditions, and bounds on the errors of composition functions. \Cref{sec:eval} describes two case studies and the experimental results. The paper wraps up with a brief discussion in \Cref{sec:discuss}.

\section{Background}
\label{sec:background}

\looseness=-1
Here we describe the preliminaries of verification and monitoring.

\subsection{Verification} \label{sec:background-verif}

\begin{definition}[System]\label{def:sys}
 A \emph{system} $\sys = (\states, \states_0, \obss, \acts, \ctrl, \dynms, \measms)$ %
consists of the following elements: 
\end{definition}
\begin{itemize}
    \item  State space \states: continuous, \edit{unbounded, finite-dimensional}, containing states \state \edit{(which include the discrete time)}%
    \item Initial states $\states_0 \subset \states$
    \item Observation space \obss, containing observations \obs %
    \item Action space \acts
    \item Controller $\ctrl: \obss \rightarrow \acts$, implemented with a neural network %
    \item Dynamics models $\dynms$: a set of functions $\dynm: \states, %
    \acts  \rightarrow \states$ %
    \item Measurement models $\measms$: a set of functions $\measm: \states%
    \rightarrow \obss$ %
\end{itemize}

A system determines a set of state traces $\statevecs(\sys)$ and a set of observation traces $\obsvecs(\sys)$ resulting from executing every combination of functions from $\dynms \times \measms$ on every initial state in $\states_0$ indefinitely. A particular realization of a system is a pair of state and observation vectors \statevec, \obsvec that occur for specific $\state_0 \in \states_0, \dynm \in \dynms$, and $\measm \in \measms$.

\edit{A \emph{safety property} $\prop$ is a Boolean predicate over traces: $\prop(\statevec) \in \{\true, \false\}$.
A property \prop is satisfied on system \sys, denoted $\sys \models \prop$, iff every trace from that system satisfies \prop: $\forall \statevec \in \statevecs(\sys),\prop(\statevec) = \true$.} \edit{Thus, this paper considers arbitrary deterministic temporal safety properties.}

\looseness=-1
A \emph{verification assumption} \edit{\assn is a restriction on the system's ``unknowns'' ---} the states and models of the system --- so, $\assn \subseteq \states \times \dynms \times \measms $. The assumption \edit{holds} on a trace $(\statevec, \obsvec)$ if for any combination of $\state_0, \dynm,$ and $ \measm$ that realizes this trace it is true that $(\state_0, \dynm,\measm) \in \assn$. %
When assumptions $\assn_1 \dots \assn_n$ are combined with a propositional logical formula \propform,  $\assn_\propform = \propform(\assn_1 \dots \assn_n)$ is also an assumption. The meaning of propositional operators $(\land, \lor,\neg,  \implies)$ is defined by the corresponding set operations (intersection for $\land$, union for $\lor$, etc). %
A system $\sys = (\states, \states_0, \obss, \acts, \ctrl, \dynms, \measms)$ \emph{under assumption} $\assn = (\states_0', \dynms'$, $\measms')$, denoted as $\sys_\assn$, is an intersection of the initial states and respective models: $ \sys_\assn = (\states, \states_0 \cap \states_0' , \obss, \acts, \ctrl, \dynms \cap \dynms', \measms \cap \measms')$.

\looseness=-1
For a given system \sys, a verification result of property \prop, denoted as $\verif_{\sys,\prop}$, is a function that maps any assumption \assn to \{ \true, \false \}. It represents the outcome of a verification effort \edit{under that assumption}, regardless of \edit{the exact} method. Value \true is assigned only if \prop was guaranteed by the verification algorithm, whereas \false is assigned in all the other cases (counterexample exists, uncertainty too high, time limit reached, etc). 
Since verification is over-approximate and exhaustive, it never produces a false safety outcome: $\verif_{\sys,\prop}(\assn) = \true \implies \sys_\assn \models \prop$. Such an assumption \assn is called \emph{sufficient} for \prop.

\subsection{Confidence Monitoring}\label{sec:backmon}

\looseness=-1
\edit{Intuitively, we want to compute the \emph{confidence} in (i.e., an estimate of the probability of) the system satisfying a safety property  \prop in the future given a prefix of observations $\obsvec$. Confidences are computed by CPS monitors in uncertain conditions, when the exact state, dynamics, and measurement model are not known. Therefore, we represent the selection of the actual system as a random sampling of  the system's unknowns --- the initial state $x_0$, dynamics $f_d$, and measurement function $f_m$ --- from some unknown distribution \dist over \states, \dynms, and \measms.}

\edit{Once we fix the distribution \dist, it induces the distribution $\dist_{\statevec, \obsvec}$ on the system's realization $(\statevec,\obsvec)$. Therefore, our monitoring goal is to estimate the probability of safety given observations up to the current moment, namely $\prob_{(\statevec,\obsvec) \sim \dist_{(\statevec,\obsvec)}}\big(\prop(\statevec) = \true \mid \obsvec_{1..n}\big)$, where $\obsvec_{1..n}$ means the first $n$ elements of \obsvec.} \edit{We pursue this goal by monitoring confidence in assumptions. Since an assumption \assn can be seen as a predicate over random $(\state_0, \dynm, \measm) \sim \dist$, its satisfaction is also random: $\assn \sim \mathcal{D}_\assn$, where $\mathcal{D}_\assn$ is induced by \dist.}

\looseness=-1
\edit{A confidence monitor $\mon: \obsvecs \rightarrow [0,1]$ for assumption \assn takes $\obsvec_{1..n}$ and outputs its estimate of $\prob_{\obsvec \sim \mathcal{D}_{\obsvec}} (\obsvec_{1..n} \in \obsvecs(\sys_\assn))$, that is, its degree of belief that the observations came from a system where \assn holds. The monitor's output, $\mon(\obsvec)$, is stochastic because it depends on \obsvec.} 
\edit{Since monitors estimate probabilities, we measure the quality of monitors using three types of calibration error with respect to \assn: }
\edit{\begin{itemize}
    \item \emph{Expected calibration error} (\ecerr):   
    $$ \ecerr(\mon, \assn) := \e_{\obsvec \sim \mathcal{D}_{\obsvec}} [| \prob_{\assn \sim \mathcal{D}_\assn}\big(\assn \mid \mon(\obsvec)\big) - \mon(\obsvec) |] $$ 
    \item \emph{Maximum calibration error} (\mcerr):
    $$ \mcerr(\mon, \assn) := \max_{p \in [0,1]} [| \prob_{\assn \sim \mathcal{D}_{\assn},\obsvec \sim \mathcal{D}_{\obsvec}}\big(\assn \mid \mon(\obsvec) = p\big) - p |] $$ 
    \item \emph{Conservative calibration %
    error (\conserr):}
    $$  \conserr(\mon, \assn) := \max_{p \in [0,1]} [p - \prob_{\assn \sim \mathcal{D}_{\assn},\obsvec \sim \mathcal{D}_{\obsvec} }\big(\assn \mid \mon(\obsvec) = p\big)] $$ 
\end{itemize}
}

\edit{\ecerr and \mcerr are widely used measures of calibration~\cite{naeini_obtaining_2015,guo17}. The concept of \conserr is novel --- we introduce it to asymmetrically quantify safety in critical systems: false alarms are safe, but missed alarms are not. \mcerr is the strictest measure of monitor quality because $\mcerr \geq \ecerr$ and $ \mcerr \geq \conserr$. When $\ecerr=0$, the monitor is \emph{calibrated in expectation}. When $\mcerr = 0$, the monitor is \emph{perfectly calibrated}. When $\conserr \leq 0$, the monitor is \emph{conservative}.}

\edit{For brevity, we will omit the distributions when they are clear from the context and refer to monitor output $\mon(\obsvec)$ as just \mon. We assume that this output is characterized by a continuous probability density $\prob(\mon)$ with finite expectation $\e[\mon]$ and variance $\var[\mon]$.}

\section{Related Work}\label{sec:related}

The \edit{research related to this paper} spans several areas: detection and estimation, aggregation of probabilities, run-time monitoring and assurance, and assumption monitoring.

Anomaly detection is a well-studied problem in  control theory and signal processing. In particular, there are multiple well-written books on sequential detection and estimation~\cite{wald04,poor13,scharf91}. %
We rely on two groups of such methods. First, filtering and parameter estimation can be used to implement the monitors that estimate states and noise parameters. Specifically, we implemented monitors using standard Monte Carlo and particle filtering~\cite{dellaert99}.
Second, monitoring model validity is related to classical change detection~\cite{willsky76} as well as more recent computational model validation methods~\cite{ozay14,prajna06,carpenter_modelguard_2021}, one of which we use to monitor assumptions on our model. Unlike classical detection methods, our work focuses on calibration~\cite{guo17}, i.e., the faithfulness of the detector to the true frequency of the underlying event. While classical model-based detectors~\cite{willsky76} can be considered ``calibrated'' by design under distributional assumptions (since probabilities can be explicitly computed), recent computational approaches~\cite{carpenter_modelguard_2021} and classical detectors under unknown distributions are in essence ``black-box'' and need to be calibrated post-hoc in order to provide performance guarantees.

Combining probability estimates is well-studied in statistics and artificial intelligence~\cite{ranjan_combining_2010,winkler_averaging_2018,sagi_ensemble_2018}. In a typical setting, such as forecast aggregation or ensemble learning, the combined methods estimate the probability of the \emph{same} underlying phenomenon. In contrast, our monitors predict fundamentally \emph{different} assumptions combined with logical operators. This setting can be interpreted as probabilistic graphical models with calibration constraints (as opposed to factor weights or conditional probabilities)~\cite{koller_probabilistic_2009}, and our product composition corresponds to the noisy-OR graphical model~\cite{pearl_probabilistic_1988}. Copulas~\cite{nelsen_introduction_2006} encode low-dimensional joint distributions with given marginals and can be used to model dependencies between verification assumptions, which we have so far assumed conditionally independent. Broadly, the literature on combining probabilities inspires the functions we use for confidence composition. %

Confidence is emerging as a key concept for expressing uncertainty in learning-enabled systems in \edit{such scenarios as detecting objects}~\cite{boursinos_improving_2020} and anticipating human motion~\cite{fisac_probabilistically_2018}. Confidences can be endowed with distributions to enable effective and general inference~\cite{shen_prediction_2018}. However, \edit{the poor} calibration of \edit{confidences} remains a major issue, especially for neural networks~\cite{guo17,van_calster_calibration_2019}. We study the calibration of black-box monitors in a safety-critical compositional setting without detailed assumptions on confidence distributions. 

Run-time monitoring is increasingly important in assurance of cyber-physical systems, with multiple run-time assurance frameworks proposed recently~\cite{desai_soter_2019,asaadi_quantifying_2020,schumann_model-based_2019,boursinos_assurance_2021,tran_decentralized_2019,mitsch_modelplex:_2014}. Some of them focus on safety\edit{-preserving} decision-making rather than accurate monitoring~\cite{desai_soter_2019,schumann_model-based_2019,cailliau_runtime_2019}. Others focus on specifying and monitoring safety properties in Linear/Metric/Signal Temporal logic with uncertainty~\cite{watanabe_runtime_2018,stoller_runtime_2012,cimatti_assumption-based_2019}, whereas we indirectly predict the satisfaction of safety properties. \edit{Yet} others provide well-calibrated confidence with non-compositional techniques such as conformal prediction~\cite{boursinos_assurance_2021} and dynamic Bayesian networks~\cite{asaadi_quantifying_2020} --- and thus can be incorporated into our framework as individual monitors. A \edit{closely} related recent framework is ReSonAte~\cite{hartsell_resonate_2021}, based on representing risks with bowtie diagrams (a counterpart of our propositional formulas) and estimating risk by using conditional distributions between system states and failures. \edit{ReSonAte's} approach is analogous to our Bayesian composition, which learns a joint distribution of monitors conditioned on assumptions --- \edit{and} naturally requires joint monitor samples or additional\edit{,} strong independence assumptions.

\looseness=-1
Assumptions have long been \edit{considered a potential cause of failures in safety-critical systems}~\cite{petroski_engineer_1992,fu_uacfinder_2020,saqui-sannes_making_2016}. %
The probabilistic and compositional formalization of assumptions is most common in frameworks for assume-guarantee reasoning and compositional verification~\cite{kwiatkowska_compositional_2013,elkader_automated_2015,ruchkin_compositional_2020}. \edit{Our paper investigates} a complementary direction of \edit{connecting} the guarantees of closed-loop neural network verification~\cite{ivanov_verisig_2021} with uncertain and imprecisely modeled run-time environments \edit{by using assumptions as an ``interface'' between the two.}
Prior works have pioneered assumption monitoring in model-based, non-deterministic settings: explicitly specified monitors~\cite{sokolsky_monitoring_2016} and monitors of proof obligations with partially observable variables~\cite{mitsch_verified_2018,cimatti_assumption-based_2019}. These model-based approaches have the advantage of verifying monitors \edit{within} the semantics of their respective models. \edit{Pursuing our vision of compositional confidence-based assurance}~\cite{ruchkin_confidence_2021}, \edit{this paper extends} confidence monitoring of assumptions to a setting where monitors do not conform to any given semantics and can exhibit unknown \edit{stochastic} behavior. %

\section{\hspace{-1.3mm}Confidence Composition Framework}\label{sec:fwk}

Our \edit{\coco} framework uses the following intuition. 
Suppose that we have monitors $\mon_1 \dots \mon_n$ for assumptions $\assn_1 \dots \assn_n$, some combination of which, \compassn, is sufficient for, and thus predictive of, safety. From $\mon_1 \dots \mon_n$, we can build a compositional monitor \compmon of $\assn_\propform$ that will estimate $\prob(\compassn)$, which is used as an indirect estimate of the chance of safety. Our composite monitor $\compmon: [0,1]^n \rightarrow [0,1]$ has form $\compfun(\mon_1 \dots \mon_n)$, where \compfun is a \emph{composition function} selected depending on \propform. Our framework formalizes the argument that if safety depends on the assumptions that have monitors with bounded \ecerr (or \conserr), %
then an appropriate compositional monitor will 
have bounded \ecerr (\conserr resp.) error with respect to the safety chance. 
This argument needs to account for model inaccuracies, overly conservative assumptions, and imperfect monitors. %

Our framework imposes certain requirements on the models, explained in the next subsection, \edit{and proceeds in three steps:} %
\begin{enumerate}
    \item Perform verification and elicit the assumptions  sufficient for safety (\Cref{sec:verifassn})
    \item Build and calibrate a confidence monitor for each assumption (\Cref{sec:confmon})
    \item Compose monitors using a composition function with desirable  bounds on the calibration error (\Cref{sec:compconf})
\end{enumerate}

In each step, we identify the framework's requirements and briefly outline how they can be achieved. \Cref{sec:endtoend} capitalizes on these requirements by providing end-to-end bounds that link the composed confidence and the true chance of safety.%

\subsection{Model Requirements}\label{sec:modelreq}

From the modeling standpoint, we distinguish two subsystems of the overarching system \sys: the unknown, true, ``real'' subsystem \sysreal %
and the modeled, known subsystem \sysmod that will undergo verification. %
Systems \sysreal and \sysmod are \textit{compatible}, i.e., they share the same \states, \obss, \acts, and \ctrl. The latter can be shared because we explicitly encode and verify the NN controller in our model. We fix some safety property \prop and verification method $\verif_{\sysmod, \prop}$ \edit{and introduce} the notion of \emph{safety relevance} between two compatible models.

\begin{definition}[\edit{Safety-relevant model}]
\label{def:safety-relevant}
A model $\sys_1$ is \emph{safety-relevant} up to a bound $e$ for a compatible model $\sys_2$ if it accounts for the safe behaviors of $\sys_2$ %
most of the time.  Formally, for a random state trace \statevec, safety property \prop, and some small $e \in [0,1]$,
$$ \prob\big( \prop( \statevec ) = \true \mid \statevec \in \statevecs(\sys_2 ) \land \statevec \not\in \statevecs(\sys_1) \big) \leq e. $$
\end{definition}

We expect the system model \sysmod to be safety-relevant for the real system \sysreal and in that case just say it is safety-relevant.  We also expect \sysmod to be \emph{verifiable}.

\begin{definition}[Verifiable model]%

A model \sysmod is \emph{verifiable} if there is a non-trivial assumption \assn (i.e., containing many states and/or models) such that all traces of $\sysmod_{\assn}$ are safe: 
    $$\exists \assn, |\assn| > 0 \land \verif_{\sysmod,\prop}(\assn) = \true.   $$ 
\end{definition}

Safety relevance intuitively means that we are unlikely to get a safe trace not represented by our model. This gives verification an opportunity to verify a system that ``explains'' a large part of \edit{truly} safe behaviors. Then, failing \edit{the} verification would  %
correspond to a low true chance of safety. %
Without safety relevance, whether verification holds may be orthogonal to whether the system is safe. 

Verifiability of a sizeable set of assumptions is important because if the model is safe only under trivial assumptions, few observed traces would satisfy \edit{them}. Then, the monitors would be forced to alarm perpetually and, hence, poorly predict safety. For instance, if we verified a system model only with zero measurement noise, a monitor would almost always invalidate this model on a real system. %
Verifiability is challenging to achieve due to the scalability and uncertainty limits of over-approximating reachability algorithms.  

There is a trade-off between safety relevance and verifiability: expanding the set of explained behaviors  %
leads to more parameters and a larger scope of the model, making it harder to verify. In the case studies, we negotiated
this trade-off by starting with simple verifiable low-dimensional models and iteratively extending their relevance while preserving their verifiability.

\subsection{Verification Assumptions}\label{sec:verifassn}

\eat{
The goal in this step is to develop verification assumptions \assn with three characteristics: 
\begin{itemize}
    \item \emph{Sufficient:}  these assumptions \assn are sufficient to obtain the verification guarantee on 
    \sysmod: 
    $$ \verif_{\sysmod,\prop}(\assn) = \true $$
   \item \emph{Safety-relevant:} The assumptions are safety-relevant: their violation leads to violating the safety with high probability.  
    $$\prob( \sysmod_{\neg \assn} \models \prop ) \leq e_2 $$
    \item \emph{Monitorable decomposition:} there exists a decomposition of \assn into $\assn_1 \dots \assn_n$ such that (i) a propositional formula \propform ensures the original assumption: $ \propform(\assn_1 \dots \assn_n) \implies \assn $, and (ii) for each sub-assumption $\assn_i$, there exists a monitor $\mon_i$ with a (preferably tight) bound on \ecerr and \conserr.  %
    
\end{itemize}
}

Assumptions are made to  constrain the modeled system \sysmod and pass safety verification. For assumption monitoring to be useful, we expect our assumptions to be \emph{sufficient} and \emph{safety-relevant}. %
\begin{definition}[Sufficient assumption]
An assumption \assn is \emph{sufficient} for property \prop if the verification of \prop on \sysmod succeeds for \assn: %
$$ \verif_{\sysmod,\prop}(\assn) = \true $$
\end{definition}

\begin{definition}[Safety-relevant assumption]
An assumption \assn applied to \sysmod is \emph{safety-relevant} up to some bound $e\in[0,1]$ if the subsystem $\sysmod_{\assn}$ is safety-relevant up to $e$ for \sysmod as per \Cref{def:safety-relevant}.

\end{definition}

\looseness=-1
\edit{The} sufficiency enables the \edit{system's} verification and is typically straightforward to achieve via grid search. We \edit{partition} the \edit{joint} space of initial states and model parameters into \edit{hypercubes}, \edit{optionally} simulate the model to \edit{quickly} rule out the unsafe regions, and then verify each remaining hypercube in parallel. The union of hypercubes where $\verif_{\sysmod,\prop} = \true$ then becomes our sufficient assumption \assn.

Notice that the safety-relevance of assumptions is defined analogously to the safety-relevance of the models, and for the same reason: we want assumption failures to correlate with safety failures. The combined safety-relevance of the models and its assumptions gives us a bound on the chance of true safety: %

\begin{theorem}[Bounded safety under failed assumptions]\label{thm:safebound}
If model \sysmod is safety-relevant up to $e_1$ and assumption \assn is %
sufficient and  safety-relevant \edit{for \sysmod} up to $e_2$, then the safety chance %
under violated assumptions is bounded by $e_1 + e_2$. I.e.,  %
for a random trace $\statevec \in \statevecs(\sysreal)$, 
$$\prob\big(\prop(\statevec) = \true \mid \statevec \not\in \statevecs(\sysmod_\assn)\big) \leq e_1 + e_2.$$
\end{theorem}
\begin{proof}
See Appendix~\ref{app:safebound}.
\end{proof}

We also expect that the sufficient and safety-relevant assumption can be monitored with the available monitors. Specifically, we expect that we can find a \emph{monitorable decomposition} of \assn into sub-assumptions $\assn_1 \dots \assn_n$ such that (i) a propositional formula \propform ensures the original assumption: $ \propform(\assn_1 \dots \assn_n) \implies \assn $, and (ii) for each sub-assumption $\assn_i$, there exists a monitor $\mon_i$ with a (preferably tight) bound on %
\ecerr and \conserr. 

The monitorable decomposition is a key idea behind this paper: while it is difficult to build a monolithic monitor for the exact and complete verification assumptions, it is possible to isolate monitorable sub-assumptions. We choose to decompose the assumption using propositional logic because the logical operators directly correspond to the set operations on the states and possible models. Notice that $\propform(\assn_1 \dots \assn_n)$ is required to imply \assn~--- not the other way around --- to ensure conservative monitoring:
$$ \prob\big(\propform(\assn_1 \dots \assn_n)\big) \leq \prob(\assn) $$ 

In practice, the choice of how to decompose assumptions depends on the available information in the  observations, the available monitoring techniques, and the scalability of monitors at run time. Often, the logical structure arises from the hazard analysis of the system. For example, if at least one of the redundant sensors functions correctly, the system can guarantee performance. This corresponds to a disjunction of the assumptions.

In our case studies, closed-loop reachability verification of hybrid systems with NN controllers relies on three categories of assumptions, suggested by \Cref{def:sys}. First, verification assumes that a system starts in \emph{initial states} from where it can avoid safety violations. Second, verification assumes that the reality is approximately described by the \emph{dynamics equations}, \edit{which are used to} propagate the reachable sets over time. Third, verification assumes that control inputs are related to the true state by a constrained set of \emph{observation models} that capture known sensor uncertainties. %
An assumption can span more than one category, and the particular decomposition depends on the specifics of the model and available monitors. %

\subsection{Confidence Monitors}\label{sec:confmon}

The goal is, for each sub-assumption $\assn_i$, to obtain a confidence monitor $\mon_i$ that estimates $\prob(\assn_i \mid \obsvec)$ with bounded errors \ecerr and \conserr. In our case studies, we built a state estimation-based monitor to determine whether the state at $t=0$ (or the current $t$) was part of the verification assumptions. Another monitor determines if the latest observations were consistent (up to some error bound) with the dynamical model under bounded observation noise. Our monitors were based on the existing detection and estimation techniques~\cite{dellaert99,carpenter_modelguard_2021}: Monte Carlo estimation using the dynamical model, particle filtering based on the dynamical model, and statistical model invalidation.

The produced monitors are often miscalibrated.
There is a trade-off between reducing \ecerr and \conserr of a monitor: a monitor with small \ecerr may sometimes overestimate the probability of an assumption holding, leading to a sizeable \conserr; on the other hand, a conservative monitor may significantly underestimate the probability and, hence, have large \ecerr.

Since we treat monitors as black boxes, we reduce their calibration errors with post-hoc calibration, which requires a validation dataset. Ideally, each monitor can use its own dataset without samples from other monitors, enabling independent development and tuning of the monitors.

We calibrate each monitor $\mon_i$ with Platt scaling~\cite{platt_probabilistic_1999}, a popular calibration technique, to produce a calibrated monitor $\mon_i'$. %
This technique is based on a linear transformation of the monitor's log-odds ($\operatorname{LO}$). For every monitor output $m$, we compute the calibrated value $m'$ as follows:
\begin{align*}
    m' = \frac{1}{1+\exp(c\operatorname{LO}(m) +d)}, && \operatorname{LO}(m) = \log(\frac{m}{1-m}),
\end{align*}
where $c$ and $d$ are calibration parameters to be determined.

To negotiate the tradeoff between \ecerr and \conserr, we fit calibration parameters $c$ and $d$ using weighted cross-entropy loss. The weight $\lambda \in [0,1]$ sets the relative importance of \conserr over \ecerr by penalizing overconfidence ($\lambda = 0.5$ in standard Platt scaling). The fitting is done on a validation dataset containing pairs of monitoring outputs $m_j$ and indicators $a_j$ of the assumption holding at the time: $\{ (m_j, a_j) \}$, and we fit over the calibrated scores $m'_j$:%
\begin{align}
     \argmin_{c,d} - \sum_j (1-\lambda) a_j \log(m'_j) + \lambda (1-a_j)\log(1-m'_j)
    \label{eq:weighted-ce}
\end{align}

\subsection{Composition of Confidence Monitors}\label{sec:compconf}

Our goal here is to build a compositional monitor \compmon for \compassn given $\mon_1 \dots \mon_n$ for $\assn_1 \dots \assn_n$ \edit{and provide bounds on its calibration error.}%

\begin{problem*}[\edit{Composition of Calibration Errors}] \leavevmode \\
\edit{Given \compassn and $\mon_1 \dots \mon_n$ calibrated to $\assn_1 \dots \assn_n$ with known bounds, find function \compfun with bounds on $\ecerr(\compmon, \compassn)$ and $\conserr(\compmon, \compassn)$. }
\end{problem*}

To solve this problem in full generality, one would need to know the joint distribution of random variables in what we call the \emph{monitoring probability space} \edit{induced by \dist}:
\begin{itemize}
    \item Bernoulli variable \safe indicating the ``safety event''  $\prop(\statevec) = \true$. We use \safe as an equivalent shorthand. %
    \item Bernoulli variables $\assn_1 \dots \assn_n,$ and $\assn_\propform$ corresponding to the satisfaction of assumptions and formula $\propform(\assn_1 \dots \assn_n)$. 
    \item Continuous variables $\mon_1 \dots \mon_n$ and \compmon. corresponding to the outputs of the monitors and their composition. 
\end{itemize}

\looseness=-1
We will not assume \edit{knowing} that joint distribution, nor shall we try to estimate it. Instead, \edit{this paper takes an early step to solving that problem by constraining} it with judicious simplifications: 
\begin{itemize}
    \item We know \edit{the} propositional formula \propform. 
    \item Assumption monitors have bounded \mcerr and \conserr. 
    \item Assumptions are conditionally independent given composite confidence, e.g., $\assn_i \indep \assn_j \mid \compmon$. 
    \item  Given their monitor, assumptions are independent of composition: $\assn_i \indep \compmon \mid \mon_i$.
    \item  Conditioning monitor variances on composition does not increase them: $\var(\mon_i \mid \compmon) \leq \var(\mon_i)$.
\end{itemize}

\looseness=-1
The first two simplifications mirror the process of independent, modular development of monitors. The conditional independence of assumptions has been approximately true in our case studies and is conservative unless the assumptions share a cause of violations. Future work can investigate other assumption dependencies. (Without any information about assumption dependencies, formula \propform is of limited use.) The fourth simplification indicates that a monitor provides all the relevant information about its assumption, and composition has none to add. The last simplification, informally, states the composite confidence is ``informative'' to monitors: knowing its value limits the likely values of individual monitors. Although this statement is difficult to prove, it has always held in our experiments. 

\looseness=-1
The above enables the rest of the section to proceed in three steps: 
\begin{enumerate}
\item Identify plausible families of functions \compfun (\Cref{sec:compfun}) %
    \item Provide bounds on $\ecerr(\compmon, \compassn)$ and $\conserr(\compmon, \compassn)$ given calibration bounds of individual monitor (\Cref{sec:compfun-bounds})
    \item Provide bounds on $\ecerr(\compmon, \safe)$ and $\conserr(\compmon, \safe)$ given bounds on  $\ecerr(\compmon, \compassn)$ and $\conserr(\compmon, \compassn)$ respectively (\Cref{sec:endtoend}) %
\end{enumerate}

In the first two steps, we focus on the monitoring probability sub-space without \safe: 
$$\prob(\assn_1 \dots \assn_n, \assn_\propform, \mon_1 \dots \mon_n, \compmon)$$

Given the complexity of this sub-space, this paper analyzes bounds on $\ecerr(\compmon, \compassn)$ and $\conserr(\compmon, \compassn)$ only for the scenario with two assumptions (which proves sufficient in the case studies):
$$\prob(\assn_1, \assn_2, \assn_\propform, \mon_1, \mon_2, \compmon)$$

In the third step, we focus on the sub-space without the individual assumptions and monitors:
$$\prob(\safe, \assn_\propform, \compmon)$$

\subsubsection{Composition Functions}\label{sec:compfun}

We identify composition functions in two steps: equivalently simplifying the expression $\prob(\compassn)$ and proposing plausible functions for conjunctions of assumptions.  

The structure of \compfun comes from the structure \propform. To determine \compfun for any $\propform(\mon_1 \dots \mon_n)$, we simplify the expression $\prob(\propform(\mon_1 \dots \mon_n))$ by converting $\propform(\mon_1 \dots \mon_n)$ to DNF and advancing\footnote{Using standard identities such as the inclusion-exclusion principle, e.g., $\prob(\assn_1 \lor \assn_2) = \prob(\assn_1) + \prob(\assn_2) - \prob(\assn_1 \land \assn_2)$ and $\prob(\neg \assn_1) = 1-\prob(A)$.} the probability operator until all probability operators are either over individual assumptions or conjunctions of assumptions. We replace marginal probabilities $\prob(\assn_i)$  with %
$\mon_i$, so it is sufficient to determine \compfun for conjunctions of assumptions --- and the rest of the expression is determined by the simplified version of $\prob(\compassn)$.

\edit{In the rest of this section, we will take initial steps towards addressing a key compositional sub-problem --- providing functions \compfun and their calibration bounds for $\prob(\assn_1 \land \assn_2)$.}

\begin{problem*}[\edit{Binary Conjunctive Composition of \ecerr/\conserr}] %
\edit{Given $\mon_1$ and $\mon_2$ calibrated to $\assn_1$ and $\assn_2$ respectively, find function \compfun with bounds on $\ecerr(\compmon, \assn_1 \land \assn_2)$ and $\conserr(\compmon, \assn_1 \land \assn_2).$ }
\end{problem*}

\edit{First}, we obtain \edit{candidate functions} \compfun via plausible restrictions of the probability space. Since these functions are difficult to compare theoretically, we leave the task of finding the best composition function for an arbitrary conjunction of assumptions for future work. We do, however, compare them experimentally in \Cref{sec:eval}.

\textbf{Product:} If one presumes that the assumptions are mutually independent and monitors are perfectly calibrated, then: %
\begin{align*}
    \prob(\assn_1 \land \dots \land \assn_n) = \prod_{i=1..n}
    \prob(\assn_i) =
   \prod_{i=1..n}
    \mon_i
\end{align*}

Thus, here $\compfun(\mon_1 \dots \mon_n) = \prod_{i=1..n}  \mon_i$.

\looseness=-1
\textbf{Weighted averaging:} If we presume that the monitors independently predict the combination of assumptions, then inverse-variance weighting minimizes the variance of the estimate~\cite{hedges_statistical_1985}:
\begin{align*}
 \prob(\assn_1 \land \dots \land \assn_n) = \sum_{i =1..n} w_i \mon_i, && w_i = \frac{1/Var(\mon_i)}{\sum_{j = 1..n} 1/Var(\mon_j) }
\end{align*}

Thus, here $\compfun(\mon_1 \dots \mon_n) = \sum_{i =1..n} w_i \mon_i, \sum_{i = 1..n} w_i = 1$.

\looseness=-1
The above two functions are modular and driven by theoretical considerations: they do not require information about the joint distribution of monitors. We also propose
two data-driven functions that are not modular: they require samples from the joint distribution of the monitors and assumptions. Due to their sample-dependent performance, we do not derive the error bounds for them. 

\textbf{Logistic regression:} 
This standard method treats the monitor values as arbitrary features (not probabilities) that linearly predict the log-odds of $\prob(\assn_1 \dots \assn_n)$.

$$ \compfun(\mon_1 \dots \mon_n) = \frac{1}{1+e^{- w_0 - \sum_{i=1..n} w_i\mon_i}} $$ 

Parameters $w_0 \dots w_n$ are fit on data from $\mon_1\dots\mon_n$ and \compassn using $\lambda$-weighted cross-entropy loss from \Cref{eq:weighted-ce}. 

\textbf{Sequential Bayes:} if we know the joint density $p(\mon_1 \dots \mon_n)$ and its conditioning on $\assn_1 \land \dots \land \assn_n$, our sequential Bayesian estimator starts with a uniform prior $\compfun_0$ at $t=0$, and at time $t+1$ is updated as follows:
{%
$$ \compfun_{t+1}(\mon_1\dots \mon_n ) = %
\edit{\compfun_{t} \cdot} \frac{ \prob\big(\mon_1 \dots \mon_n \mid \assn_1 \land \dots \land \assn_n\big)~ }{\prob(\mon_1 \dots \mon_n )} $$
}

\subsubsection{Error bounds for composition functions}\label{sec:compfun-bounds}

First, we note a useful lemma that links assumption probabilities to monitor expectations. Then we bound \ecerr for the product and weighted averaging composition, and finally bound \conserr for the product composition. %

\begin{lemma}[Conditional Assumption Bounds]%
\leavevmode \\
\label{lem:prob-assump}
For any composition function \compfun, if:
\begin{align*}
    &\mcerr(\mon_1, \assn_1) \leq e_1, \quad \mcerr(\mon_2, \assn_2) \leq e_2,\\ 
    &%
    \assn_1 \indep \compmon \mid \mon_1,\quad  \assn_2 \indep \compmon \mid \mon_2 %
\end{align*}
Then:
\begin{align*}
    \e[\mon_1 \mid \compmon] -e_1 &\leq \prob(\assn_1 \mid \compmon)   \leq   \e[\mon_1 \mid \compmon]+e_1 %
    \\
    \e[\mon_2 \mid \compmon] -e_2 &\leq \prob(\assn_2 \mid \compmon)   \leq   \e[\mon_2 \mid \compmon]+e_2 %
\end{align*}

\end{lemma}
\begin{proof}
See Appendix~\ref{app:prob-assump}.
\end{proof}
\begin{theorem}[\edit{\ecerr bound} for product composition] %
\label{th:ece-product}
If:
\begin{align*}
    &\compmon = \mon_1 \mon_2,~\mcerr(\mon_1, \assn_1) \leq e_1,~ \mcerr(\mon_2, \assn_2) \leq e_2,\\ 
    &\assn_1 \indep \assn_2 \mid \compmon,~ \assn_1 \indep \compmon \mid \mon_1,~ \assn_2 \indep \compmon \mid \mon_2,  \\
    &\var(\mon_1 \mid \compmon) \leq \var(\mon_1),~ \var(\mon_2 \mid \compmon) \leq \var(\mon_2)
\end{align*}
Then:
\begin{align*}
&\ecerr(\mon_1\mon_2, \assn_1 \land \assn_2) \\
&\quad \leq \max [4e_1e_2, \sqrt{\var[\mon_1]\var[\mon_2]} + e_1 + e_2 + e_1e_2] \\
\end{align*}
\end{theorem}
\begin{proof}
See Appendix~\ref{app:ece-product}.
\end{proof}
The above bound is fairly narrow if the monitors are well-calibrated and have low variance. This suggests that product composition may perform well in practice.

\begin{theorem}[\edit{\ecerr bound} for weighted averaging comp.]
\label{th:ece-weighted}
If:
\begin{align*}
    &\compmon = w_1\mon_1 +w_2 \mon_2,\quad w_1+w_2 = 1,\\ 
    &\mcerr(\mon_1, \assn_1) \leq e_1,\quad \mcerr(\mon_2, \assn_2) \leq e_2,\\ &\assn_1 \indep \assn_2 \mid \compmon,\quad \assn_1 \indep \compmon \mid \mon_1,\quad \assn_2 \indep \compmon \mid \mon_2
\end{align*}
Then:
\begin{align*}
&\ecerr(w_1\mon_1 + w_2\mon_2, \assn_1 \land \assn_2) \\ 
&\quad \leq \max[e_1 + e_2 + e_1e_2,~ \max[w_1,w_2] + e_1 +e_2-e_1e_2 ]
\end{align*}
\end{theorem}
\begin{proof}
See Appendix~\ref{app:ece-weighted}.
\end{proof}

Notice that the above \ecerr bound for the averaging composition is lower-bounded by 0.5. This reflects that under our precondition, the proof is required to bound  $w_1 \mon_1 + w_2 \mon_2 - \mon_1 \mon_2$, which can reach values at least as great as 0.5. This reflects the mismatch between additive composition and conditionally independent assumptions, and we do not expect it to perform well in our case studies. Finding  different conditions under which this \ecerr bound can be tightened remains for future work.  

\begin{theorem}[Conservatism bound for product comp.]\label{th:cce-product}
If:
\begin{align*}
    &\compmon = \mon_1\mon_2, \quad
    \assn_1 \indep \assn_2 \mid \compmon,\quad \assn_1 \indep \compmon \mid \mon_1,\quad \assn_2 \indep \compmon \mid \mon_2, \\
    &\mcerr(\mon_1, \assn_1) \leq e_1,\quad \mcerr(\mon_2, \assn_2) \leq e_2,
\end{align*}
Then:
\begin{align*}
\prob(\assn_1, \assn_2 \mid \compmon=x) \geq \max[0, x - e_1] \max[0, x - e_2]
\end{align*}

\end{theorem}
\begin{proof}
See Appendix~\ref{app:cce-product}.
\end{proof}

This theorem shows that low-confidence outputs of product composition should not be trusted to be conservative. It also motivates another composition function that we expect to be more conservative:

\textbf{Power Product:} 
\begin{align*}
    \prob(\assn_1 \land \dots \land \assn_n) =
   \prod_{i=1..n}
    (\mon_i)^n
\end{align*}

\begin{corollary}[\edit{ \conserr bound} for product composition]\label{cor:cce-prod}
\leavevmode\\
Under conditions of \Cref{th:cce-product}, $\conserr(\mon_1\mon_2, \assn_1 \land \assn_2) $ is bounded by 
\[ 
 \begin{cases}
    \left(e_1^2 - 2e_1 (-1 + e_2) + (1 + e_2)^2\right)/4, & \text{if } (1+e_1+e_2)/2 \in [0,1] \\
    e_1 + e_2 -e_1e_2, & \text{otherwise} 
  \end{cases}
\]

\end{corollary}
\begin{proof}

\begin{align*}
\conserr(\mon_1\mon_2) &= \max_{x\in[0,1]} [x -   \prob(\assn_1 \land \assn_2 \mid \compmon = x) ] \\
&\leq \max_{x\in[0,1]} [x- \max[0, x   -e_1]\max[0,x-e_2] ]
\end{align*}

On interval $x \in [0, \max[e_1,e_2]]$, the maximum is achieved on $x=\max[e_1,e_2]$ and equal to $\max[e_1,e_2]$.

On interval $x \in [\max[e_1,e_2] ,1]$, we have
$f(x) = x - (x-e_1)(x-e_2) = -x^2 + (1+e_1+e_2)x - e_1 e_2$, which is a quadratic function with $f''(x) <0$. Therefore, its maximum on $[0,1]$ may be achieved only in $x=\max[e_1,e_2]$, $x=1$, or when $f'(x) = 0$. The former case \edit{is above}, leaving the other two.  $f(1) = e_1 + e_2 -e_1e_2$, which can be shown to be greater than $\max[e_1, e_2]$, eliminating that bound.

Solving $f'(x_0) = 0$ for $x_0$, we get: 
\begin{align*}
    x_0 & = ( (1+e_1+e_2)/2; 
    f(x_0) = \left(e_1^2 - 2e_1 (-1 + e_2) + (1 + e2)^2 \right)/4
\end{align*}

Thus, if $(1+e_1+e_2)/2 \in [0,1]$, we get the above bound, and $e_1 + e_2 -e_1e_2$ otherwise.

\end{proof}

\looseness=-1
When proving \Cref{th:cce-product}, we took advantage of a property of product composition: \edit{the composed confidence is a lower bound of each monitor confidence.} This property is not present in \edit{the sum-based compositions} (averaging, logistic regression): given a composition value, one of the monitors may be arbitrarily small, and so the probability of the (conditionally independent) assumptions can be arbitrarily small as well. Hence, no general conservatism bound can be provided for such compositions under our preconditions. It may be possible provide it in special cases of assumption dependencies or monitor distributions, which remain for future work. 

\subsection{End-to-End Bound on Calibration Error} \label{sec:endtoend}

Finally, we get to tie the \edit{\coco} framework together with the results from \Cref{sec:modelreq,sec:verifassn,sec:confmon,sec:compconf} and prove a \emph{key result of this paper}: when the requirements of our framework are satisfied, compositional monitors have guaranteed upper bounds on \ecerr and \conserr. The proofs can be found in \Cref{app:ece-comp-bounds,app:cce-comp-bounds}.

\begin{theorem}[End-to-end bound on \ecerr]
\label{th:ece-comp-bounds}
If the model is safety-relevant up to $e_1$, assumptions \compassn are sufficient and  safety-relevant up to $e_2$, and monitor $\compmon$ is calibrated to \compassn with $\ecerr(\compmon,\compassn) \leq e_3$, %
then it is calibrated to \safe with bounded $\ecerr(\compmon,\safe)$:  $$\ecerr(\compmon, \safe) \leq e_1 + e_2 + e_3$$
\end{theorem}

\begin{theorem}[End-to-end bound on \conserr]
\label{th:cce-comp-bounds}
If \compassn is a sufficient assumption and $\conserr(\compmon, \compassn) \leq e$, then $\conserr(\compmon, \safe) \leq e$.
\end{theorem}

Notice \edit{that} the preconditions for the \conserr bound are weaker due to relying on formal verification with strong guarantees.

\begin{table*}[tbh]
{\scriptsize
    \centering
    \resizebox{\textwidth}{!}{
    \begin{tabular}{p{1.35cm} c || c c c c c |c  c c c c c c |}
  &   & \multicolumn{5}{c}{\small Mountain Car Case Study}     &    \multicolumn{5}{c}{\small UUV Case Study}    \\ \hline
         \addstackgap[4pt]{\small Monitor} & {\small Predicts} &  {\small \eecerr} & {\small \emcerr} & {\small \econserr} & {\small \ebrier} & {\small \eauc}  & {\small \eecerr} & {\small \emcerr} & {\small \econserr} & {\small \ebrier} & {\small \eauc} \\ \hhline{============}
         \multirow{2}{*}{\small$\mon_1$}          & $ \assn_1$ & $ 0.021 \pm 0.004 $ & $0.157 \pm 0.04$ & $0.152 \pm 0.04$ & $0.043 \pm 0.002$ & $0.987 \pm 0.001$ & $ 0.1 \pm 0.02$ & $0.41 \pm 0.15$ & $0.26 \pm 0.15$ & $0.176  \pm  0.01$ & $0.829 \pm 0.009$    \\
         & \safe & $0.285 \pm 0.01$ & $0.456 \pm 0.02$ & $0.456 \pm 0.02$ & $0.313 \pm 0.01$ & $0.699 \pm 0.01$ & $0.111  \pm 0.04$ & $0.399 \pm 0.14$ & $0.316 \pm 0.12$ & $ 0.202 \pm  0.02$ &$0.796 \pm 0.01$  \\ \hline
         \multirow{2}{*}{\small$\mon_2$}          & $ \assn_2$ & $0.157 \pm 0.02$ & $ 0.322 \pm 0.04$ & $0.278 \pm 0.02$ & $0.225 \pm 0.004$ & $0.764 \pm 0.002$ & $0.197  \pm 0.06$ & $0.317 \pm 0.1$ & $0.296 \pm 0.11$ & $0.224  \pm 0.04 $ & $0.676 \pm 0.02$    \\
         & \safe & $0.241 \pm 0.02$ & $0.436 \pm 0.01$ & $0.436 \pm 0.01$ & $0.307 \pm 0.004$ & $0.674 \pm 0.007$ & $0.239  \pm 0.06$ & $0.327 \pm 0.1$ & $0.317 \pm 0.1$ & $0.261  \pm 0.04 $ &$0.659 \pm 0.02$  \\ \hlineB{1.6} %
         \multirow{2}{*}{\small$\mon_1\mon_2$} & \compassn & $ 0.087 \pm 0.01 $  & $ \mathbf{0.207 \pm 0.01} $& $ 0.207 \pm 0.01$& $ \mathbf{0.132 \pm 0.003}$ & $ \mathbf{0.887 \pm 0.004}$ & $0.109 \pm 0.02$ & $0.343  \pm 0.18$ & $0.196 \pm 0.12$  & $0.184 \pm 0.008$  & $0.82 \pm 0.02$   \\
         & \safe & $\mathbf{0.129 \pm 0.007}$ & $\mathbf{0.280 \pm 0.04}$ & $0.18 \pm 0.01$ & $\mathbf{0.202 \pm 0.004}$ & $\mathbf{0.784 \pm 0.007}$ & $0.107  \pm 0.03$ & $\mathbf{0.343 \pm 0.18}$ & $0.184  \pm 0.12$ & $0.182  \pm  0.007$ &$ 0.821 \pm 0.009$  \\ \hline
         \multirow{2}{*}{\small$w_1\mon_1+w_2\mon_2$} & \compassn&  $0.349 \pm 0.02$& $0.659 \pm 0.02$ &$0.659 \pm 0.02$ & $0.266 \pm 0.01$& $0.811 \pm 0.003$ & $0.212 \pm 0.06$  & $0.428 \pm 0.1$  & $0.428  \pm 0.1$ & $0.238 \pm 0.03$  & $0.813 \pm 0.009$  \\
          & \safe& $0.223 \pm 0.01$& $0.467 \pm 0.02$ & $0.467 \pm 0.02$& $0.244 \pm 0.006$& $0.742 \pm 0.004$ & $0.188  \pm 0.06$ & $ 0.417 \pm 0.1$ &$ 0.417  \pm 0.1$  &$0.227 \pm  0.03$  & $0.807 \pm 0.009$ \\ \hline
          \multirow{2}{*}{\small$(\mon_1\mon_2)^2$} & \compassn& $0.092 \pm 0.009$ & $0.210 \pm 0.02$ & $0.204 \pm 0.02$ & $\mathbf{0.132 \pm 0.004}$ & $\mathbf{0.887 \pm 0.004}$ &  $0.218  \pm 0.07$ & $0.342 \pm 0.05$ & $\mathbf{-0.038 \pm 0.08}$ &$0.226  \pm 0.03$  & $0.82 \pm 0.009$ \\
          & \safe& $0.213 \pm 0.01$& $0.428 \pm 0.05$ & $0.175 \pm 0.01$& $0.234 \pm 0.008$& $\mathbf{0.784 \pm 0.007}$ & $0.239  \pm 0.07$ & $0.386 \pm 0.06$  & $\mathbf{-0.044  \pm 0.08}$ & $0.235 \pm 0.03$  & $0.821 \pm 0.009$  \\ \hline
          \multirow{2}{*}{\small LogReg($\mon_1, \mon_2$)} & \compassn& $\mathbf{0.049 \pm 0.006}$ & $0.245 \pm 0.04$& $\mathbf{0.108 \pm 0.01}$ &  $\mathbf{0.13 \pm 0.003}$& $0.867 \pm 0.006$ &  $\mathbf{0.07  \pm 0.03}$ &$\mathbf{0.331  \pm 0.18}$ &$0.155 \pm 0.08$  &$0.173  \pm 0.006$  &  $0.829 \pm 0.008$ \\
          & \safe& $\mathbf{0.129 \pm 0.02}$&  $0.294 \pm 0.03$& $\mathbf{-0.018 \pm 0.03}$ & $0.212 \pm 0.007$ & $0.764 \pm 0.008$ & $\mathbf{0.079 \pm 0.03}$  & $0.402 \pm  0.18$  & $0.143  \pm 0.08$ & $0.173 \pm 0.006$  &  $0.829 \pm 0.009$  \\ \hline
           \multirow{2}{*}{\small Bayes($\mon_1, \mon_2$)} & \compassn& $0.285 \pm 0.02 $&$0.679 \pm 0.07$ & $0.679 \pm 0.07$ &$0.286 \pm 0.02$ & $\mathbf{0.886 \pm 0.02}$ &  $0.141 \pm 0.02$  & $0.438  \pm 0.07$ & $0.27  \pm 0.11$ & $\mathbf{0.155  \pm 0.02}$ & $\mathbf{0.914 \pm 0.01}$ \\
          & \safe& $0.328 \pm 0.01$& $0.572 \pm 0.06$& $0.572 \pm 0.06$ & $0.331 \pm 0.01$ & $0.76 \pm 0.02$ & $0.14  \pm 0.02$ & $0.445  \pm 0.07$ & $0.22  \pm 0.09$ & $\mathbf{0.153 \pm 0.02}$   &  $\mathbf{0.917  \pm 0.02}$
    \end{tabular}}
    \caption{Average monitor performance across 20 cross-validation runs with neutrally-weighed calibration  ($\lambda=0.5$)}
    \label{tab:equal-results}
    }
\end{table*}

\begin{table*}[tbh]
{\scriptsize
    \centering
    \resizebox{\textwidth}{!}{
    \begin{tabular}{p{1.35cm} c || c c c c c |c  c c c c c c |}
  &   & \multicolumn{5}{c}{\small Mountain Car Case Study}     &    \multicolumn{5}{c}{\small UUV Case Study}    \\ \hline
         \addstackgap[4pt]{\small Monitor} & {\small Predicts} &  {\small \eecerr} & {\small \emcerr} & {\small \econserr} & {\small \ebrier} & {\small \eauc}  & {\small \eecerr} & {\small \emcerr} & {\small \econserr} & {\small \ebrier} & {\small \eauc} \\ \hhline{============}
         \multirow{2}{*}{\small$\mon_1$}          & $ \assn_1$ & $ 0.058 \pm 0.006  $ & $0.405 \pm 0.03$ & $-0.002 \pm 0.002$ & $0.057 \pm 0.004$ & $0.987 \pm 0.001$ & $0.251  \pm 0.05 $ & $0.506 \pm  0.16$ & $-0.039 \pm 0.02 $ & $0.236 \pm 0.02  $ & ${0.825 \pm 0.01} $    \\
         & \safe & $0.315  \pm  0.009$ & $0.465 \pm 0.01$ & $0.465 \pm 0.01$ & $0.329 \pm 0.009$ & $0.693 \pm 0.01$ & $0.191  \pm 0.05$ & $0.449 \pm 0.18$ & $0.002 \pm 0.06$ & $0.228 \pm 0.02 $ & $0.793 \pm 0.02$    \\ \hline
         \multirow{2}{*}{\small$\mon_2$}          & $ \assn_2$ & $0.195 \pm  0.008$ & $0.548 \pm 0.07$ & $0.241 \pm 0.009$ & $0.236 \pm 0.002$ & $0.764  \pm 0.002$ & $0.096  \pm 0.03$ & $0.363 \pm 0.28$ & $0.099 \pm 0.03$ & $0.193 \pm 0.006 $ & $0.671 \pm 0.01$    \\
         & \safe & $0.274  \pm  0.008$ & $0.475 \pm 0.05$ & $0.437 \pm 0.01$ & $0.316 \pm 0.005$ & $0.67 \pm 0.007$ & $0.125  \pm 0.03$ & $0.337 \pm 0.2$ & $0.162 \pm 0.03$ & $0.22 \pm 0.006 $ & $0.653 \pm 0.02$    \\ \hlineB{1.6} %
         \multirow{2}{*}{\small$\mon_1\mon_2$} & \compassn & $ \mathbf{0.102 \pm 0.008 }$  & $ \mathbf{0.231 \pm 0.02} $& $ 0.203 \pm 0.009$& $\mathbf{ 0.137 \pm 0.004}$ & $ \mathbf{0.881 \pm 0.004}$ & $0.287 \pm 0.05$ & $0.429  \pm 0.05$ & $-0.076 \pm 0.02$  & $0.277\pm 0.03$  & $0.817 \pm 0.01$   \\
         & \safe & $0.223 \pm 0.008$ & $0.486 \pm 0.04$ & $0.176 \pm 0.008$ & $0.242 \pm 0.006$ & $0.779 \pm 0.008$ & $0.31  \pm 0.05$ & $0.46 \pm 0.07$ & $-0.08  \pm 0.02$ & $0.277  \pm  0.03$ &$ 0.818 \pm 0.01$  \\ \hline
         \multirow{2}{*}{\small$w_1\mon_1+w_2\mon_2$} & \compassn&  $0.229 \pm 0.008$& $0.363 \pm 0.009$ &$0.363 \pm 0.009$ & $0.197 \pm 0.002$& $0.826 \pm 0.01$ & $\mathbf{0.109 \pm 0.03}$  & $\mathbf{0.272 \pm 0.11}$  & $0.072  \pm 0.09$ & $\mathbf{0.185 \pm 0.007}$  & $\mathbf{0.825 \pm 0.01}$  \\
          & \safe& $\mathbf{0.172 \pm 0.006}$& $\mathbf{0.308 \pm 0.04}$ & $0.228 \pm 0.01$& $\mathbf{0.22 \pm 0.003}$& $0.749 \pm 0.01$ & $\mathbf{0.12 \pm 0.04}$ & $ \mathbf{0.281 \pm 0.11}$ &$ 0.051  \pm 0.08$  &$\mathbf{0.186 \pm  0.009}$  & $\mathbf{0.827 \pm 0.01}$ \\ \hline
          \multirow{2}{*}{\small$(\mon_1\mon_2)^2$} & \compassn& $0.151 \pm 0.006$ & $0.431 \pm 0.04$ & $0.197 \pm 0.01$ & $0.157 \pm 0.005$ & $\mathbf{0.881 \pm 0.004}$ &  $0.443  \pm 0.05$ & $0.668 \pm 0.05$ & $\mathbf{-0.162 \pm 0.03}$ &$0.394  \pm 0.04$  & $0.817 \pm 0.01$ \\
          & \safe& $0.273 \pm 0.007$& $0.614 \pm 0.02$ & $0.169 \pm 0.009$& $0.276 \pm 0.006$& $0.779 \pm 0.008$ & $0.466 \pm 0.05$ & $0.67 \pm 0.05$  & $\mathbf{-0.164  \pm 0.03}$ & $0.415 \pm 0.04$  & $0.818 \pm 0.01$  \\ \hline
          \multirow{2}{*}{\small LogReg($\mon_1, \mon_2$)} & \compassn& $0.144 \pm 0.01$& $0.447 \pm 0.02$& $\mathbf{-0.059 \pm 0.008}$ &  $0.16 \pm 0.006$& $0.868 \pm 0.005$ &  $0.235  \pm 0.05$ &$0.472  \pm 0.13$ &$-0.044 \pm 0.04$  &$0.231  \pm 0.02$  &  $\mathbf{0.826 \pm 0.009}$ \\
          & \safe& $0.276 \pm 0.009$&  $0.481 \pm 0.02$& $\mathbf{-0.237 \pm 0.01}$ & $0.275 \pm 0.007$ & $0.761 \pm 0.008$ & $0.258 \pm 0.05$  & $0.569 \pm  0.1$  & $-0.048  \pm 0.04$ & $0.243 \pm 0.02$  &  $\mathbf{0.827 \pm 0.01}$
    \end{tabular}}
    \caption{Average monitor performance across 20 cross-validation runs with conservatively-weighed calibration  ($\lambda=0.8$)}
    \label{tab:cons-results}
    }
\end{table*}

\section{Evaluation}\label{sec:eval} 

\looseness=-1
The goal of our case studies is to evaluate the usefulness of \edit{the \coco} framework as a whole (``does assumption monitoring predict safety violations?''), the usefulness of confidence composition (``\edit{do compositions outperform} their constituents?''), and our ability to improve conservatism (``\edit{how} to reduce the \conserr of \edit{compositions}?''). 

\looseness=-1
We perform two case studies: a mountain car getting up a hill and an underwater vehicle tracking a pipeline. Each system has two verification assumptions and two monitors. The studies differ in several ways to show the flexibility of \edit{\coco}: the safety properties are somewhat different (eventually vs always), the true model is unknown to us in the second case study, initial-state assumptions are combined differently with measurement assumptions, the state assumption is evaluated at different times ($t=0$ and current $t$), and the state estimation monitors use different techniques. %

\looseness=-1
Our plan is to execute each system and collect, for each monitor, a dataset of $N$  monitor outputs, $\emon = \{m_1 \dots m_N\}$, binary satisfactions of the respective monitored assumption, $\eassn = \{a_1 \dots a_N\}$, and true eventual binary safety outcomes (the chance of which the monitor predicts indirectly), $\esafe = \{\phi_1 \dots \phi_N\}$. 

Our analysis will measure the binned approximations of the calibration errors from \Cref{sec:backmon} on a uniform binning $B_1 \dots B_K$ of $[0,1]$ into $K=10$ confidence bins. We compare the average confidence within each bin, $\operatorname{conf}(B_k) := \frac{1}{|B_k|} \sum_{i \in B_k} m_i$, with the rate of assumption occurrence in that bin, $\operatorname{occ}(B_k) := \frac{1}{|B_k|} \sum_{i \in B_k} a_i $.

\begin{itemize}
    \item \emph{Estimated expected calibration error:} (\eecerr)
    $$ \eecerr(\emon, \eassn) := \sum_{k=1}^K \frac{|B_k|}{N} |\operatorname{occ}(B_k) - \operatorname{conf}(B_k)|$$

    \item \emph{Estimated maximum calibration error:} (\emcerr)
     $$ \emcerr(\emon, \eassn) := \max_{k \in K} |\operatorname{occ}(B_k) - \operatorname{conf}(B_k)| $$
        
    \item \emph{Estimated conservatism error:} (\econserr)
    $$ \econserr(\emon, \eassn) := \max_{k \in K} [\operatorname{conf}(B_k) - \operatorname{occ}(B_k)] $$

\end{itemize}

To evaluate how calibrated \emon is to safety, \eassn is replaced with \esafe in the above definitions.

Calibration should not be evaluated in isolation from accuracy-related measures; otherwise, a monitor could ``cheat'' by always outputting an estimate of average probability ---  and thus give up its ability to discriminate the outcomes. So, in addition to calibration measures, we will provide two measures of accuracy:
\begin{itemize}
    \looseness=-1
    \item \emph{Estimated Brier Score} (\ebrier) is a classic scoring rule for probability predictions~\cite{ranjan_combining_2010,winkler_averaging_2018}, \edit{a.k.a.} the mean squared error:
    \hspace{-5mm}
    $$ \ebrier(\emon, \eassn) := \frac{1}{N} \sum_{i=1}^N (m_i - a_i)^2 $$
    \hspace{-5mm}
    \item \emph{Area under Curve} (\eauc) of the trade-off (ROC) curve between the true positive and false positive rates, which manifests when \emon is thresholded by every number between 0 to 1. It is used to measure a classifier's discrimination ability. %
\end{itemize}

\edit{The case study data and the source code for its analysis are available at \url{https://github.com/bisc/coco-case-studies}.}

\subsection{Mountain Car}

\edit{The} \emph{mountain car} (MC)~\cite{moore91} is a standard reinforcement learning benchmark where the task is to drive an underpowered car up a hill from a valley. The controller needs to first drive the car up the opposite hill so as to gather enough speed. Formally, the car has two continuous states, position and velocity, both one-dimensional in the horizontal direction, \statevec~:= (position $p$, velocity $v$).

The car is considered safe if it gets to the top of the hill in 110 steps: \prop~:= $t \geq 110 \implies p \geq 0.45$. %
The dynamics \dynms is as follows:
\begin{align*}
& p_{k+1} = p_k + v_{k}, &&
v_{k+1} = v_k + 0.0015u_k - z*cos(3p_k),
\end{align*}
\looseness=-1
where $u_k \edit{\in [-1,1]}$ is the controller's output, and $z$ is the hill steepness, \edit{sampled uniformly from} two values: $\{0.0025, 0.0035\}$. %
Initial position $p_0$ is \edit{sampled} uniformly from $[-0.6, -0.4]$, and $v_0 = 0$.  %

In our extension of the classic mountain car, noisy measurements \obsvec~:= (estimated position $\hat{p}$, estimated velocity $\hat{v}$) are obtained from measurement models \measms in which driving faster makes localization more difficult, and being on a hillside biases the velocity estimates. This model uses noise parameters $c$ and $d$ chosen uniformly from $[-1, 1]$ and $[-0.01, 0.02]$, respectively:
\begin{align*}
   & \hat{p}_k = p_k + cv_{k}, &&  \hat{v}_k = v_k + dp_{k}
\end{align*}

We use a NN controller that was trained and verified in the related work~\cite{ivanov19}. We extended its verification under two assumptions: $\assn_1$ encodes the relation between the initial states and noise parameters where the verification can guarantee as a predicate over $p_0, c$, and $d$; $\assn_2$ expects the execution to follow \dynms above with $z=0.025$ (less steep hill) and noise $c$ and $d$ from the intervals above, up to a certain error bound. The car may fail due to violating either assumption, leading to $\compassn = \assn_1 \land \assn_2$. %

Monitoring $\assn_1$ is performed by an initial Monte-Carlo sample of triples $(p_0, c, d)$ and gradually inferring their weights based on \dynms and the observations. $\mon_1$ outputs the weight fraction of the samples that satisfy the predicate from $\assn_1$. Monitoring $\assn_2$ is performed by statistically testing the consistency between \dynms, \measms, and a trace of last 6 observations using an existing tool ModelGuard~\cite{carpenter_modelguard_2021}. The confidence $\mon_2$ is either 1 when the model matches the execution or the percentage of the model parameter space that was explored and found to be inconsistent with the execution. %

In data collection, we uniformly sample initial states, noise parameters, hill steepness $z$, and add white process noise $(\mathcal{N}(0,0.001)$, $\mathcal{N}(0,0.0001))$ to introduce a slight mismatch between our model \sysmod and the ``real'' system \sysreal. We collected 2002 MC executions with $N=196449$ samples total.  

\subsection{Unmanned Underwater Vehicle}

The second case study, is an \emph{unmanned underwater vehicle} (UUV) based on \edit{a challenge problem from the DARPA Assured Autonomy program. The UUV follows an underwater pipeline and inspects it for cracks. Here, the ``real" system \sysreal is implemented with a high-fidelity UUV simulator} based on the Robot Operating System~\cite{manhaes16}.  We use a linearized identified dynamics model \dynms of the UUV.  %
The states are \statevec~ := (two-dimensional position $p_x$, $p_y$, heading $\theta$, velocity $v$, and 4 digital variables from system identification), and the pipe coincides with $x$-axis. %
The measurement \obsvec~:= (heading to \edit{the} pipe $\hat{\theta}$, range to \edit{the} pipe $\hat{r}$) contains a sensor estimate of the heading $\theta$, which is the angle formed between the UUV's direction and the positive $x$-axis. It also contains the range measurement $\hat{r}$ which is the distance to the horizontal pipeline at $y=0$, perpendicular to the heading. The measurements are computed as:
\begin{align*}
\vspace{-3mm}
    & \hat{r}_k = \frac{p_{y,k}}{cos(\theta_k)},   && \hat{\theta}_k = \theta_k + D\mathbf{\mathit{u}}_k,
    \vspace{-3mm}
\end{align*}
where the coefficient matrix $D = \begin{bmatrix}0 & 0 & \dots & d\end{bmatrix}^T$ extracts the controlled turn angle and multiplies with a noise parameter $d \in [-0.1, 1.5]$. Intuitively, $d$ approximates heading estimation delays during turns, improving our model's safety-relevance. 

\edit{We consider the UUV safe} at any time $\hat{t}$ if for the next 30 seconds its distance to pipe is within the side-looking sonar range (between 10m and 50m):  %
\prop~:= $t \leq \hat{t} + 30 \implies 10 \leq p_y \leq 50  $. %
Given these measurements, we train a NN controller running at 0.5 Hz to follow the pipe using the TD3 reinforcement learning algorithm~\cite{fujimoto18}.

\looseness=-1
Assumption $\assn_1$ is a predicate over $p_y$ and $\theta$ where the verification succeeded. \edit{We monitor it with a particle filter, which propagates particles} ($p_y, \theta)$ over time. $\mon_1$ outputs the weight fraction of the current particles within the $\assn_1$. $\assn_2$ expects the system to behave consistently with our \dynms and \measms with $d \in [-0.1, 1.5]$ based on 4 latest observation steps, and $\mon_2$ is \edit{analogous to $\mon_2$ for MC}. A conjunction of these assumptions proved sufficient for safety: $\compassn = \assn_1 \land \assn_2$. 

\looseness=-1
Our scenario is the UUV heading towards the pipe and needing to make a sharp left turn before $p_y < 10$. We sample $p_{y,0}$ between 12m and 22m and $\theta_0$ between 5 and 25 degrees, some of which violate $\assn_1$. We also introduce a $33\%$ chance of a fin getting stuck and limiting the turn rate, which violates $\assn_2$ and makes the UUV less likely to maintain safety. We collected 194 UUV executions and \edit{evaluated the safety predictions} in the first 20 seconds of each ($N=3880$).

\subsection{Results}

We experiment with two calibration settings: \edit{neutral} ($\lambda=0.5$) and \edit{conservative} ($\lambda=0.8$). In each, the monitors are evaluated with 50-50 cross-validation: \edit{tuned} on a randomly chosen half of the data and \edit{tested} on the rest. On the validation set, the monitors $\emon_1$, $\emon_2$ are \emph{individually} calibrated with Platt scaling and composed into product, weighted average (weights set inverse to the post-calibration variance), and the product-squared (the two-monitor version of power product). The other two methods use the \emph{joint} monitor-assumption data: logistic regression fits monitor outputs to \compassn, and  Bayes estimates $\prob(\mon_1, \mon_2 \mid \compassn)$ with histograms. 

\looseness=-1
We present the monitor evaluations in \Cref{tab:equal-results} ($\lambda = 0.5$) and \Cref{tab:cons-results} ($\lambda=0.8$). Each table contains the means and standard deviations of monitor errors and scores after 20 cross-validations. Notice two caveats: (i) $\mon_1$ and $\mon_2$ predict $\assn_1$ and $\assn_2$, not \compassn, (ii) since Bayes is not affected by $\lambda$, it is omitted from \Cref{tab:cons-results}. What follows is our observations and interpretations, mostly based on \Cref{tab:equal-results}. 

\textbf{Framework predicts safety}: compositional monitors show \eauc above $0.7$ in the MC case and above $0.8$ in the UUV case. \eecerr stays within $0.1-0.2$, \edit{explained in part} by the difficulty of calibrating $\mon_2$, which tends to take extreme values. \edit{The good composite calibration is supported by} the high safety relevance of our assumptions: on traces with $\neg \compassn$, \edit{the} safety chances \edit{are} $18\%$ for the MC and $6\%$ \edit{for} the UUV. Brier scores are relatively high because our ``true probabilities'' are 0s and 1s, which penalized the predictions in the middle of $[0,1]$ --- even when well-calibrated and discriminative.

\textbf{Compositions outperform  constituents:}  \edit{when predicting safety (\safe)}, all compositional monitors have a higher \eauc than $\mon_1$ and $\mon_2$, and most have a higher \eecerr. We note that compositions have higher uncertainty, which drives up their \ebrier scores.

\textbf{Data-driven compositions outperform the non-data-driven:} logistic regression shows the best \eecerr on both case studies, and Bayes dominates the \eauc \edit{in most cases}. This outcome is expected given their information advantage. Among the non-data-driven compositions, as predicted by our theorems, the product outperforms the weighted average across all metrics in \Cref{tab:equal-results}.

\textbf{Compositions are tunable for conservatism:} as per \Cref{tab:cons-results}, the product and product-squared compositions improved their \econserr with $\lambda=0.8$ in the UUV case, as expected. Logistic regression got more conservative in both cases studies, which motivates \edit{trying to mimic} its benefits without joint data. Surprisingly,  the product and product-squared did not respond to conservative tuning in the MC case. Our investigation showed that although $\mon_2$ got more conservative on average, its predictions with $\geq$90\% confidence got \emph{less} conservative. This suggests future work in  calibration techniques with conservative guarantees \edit{and using confidence monitoring as an early warning rather than a final arbiter of safety.} Product-squared failed to improve conservatism for the MC due to a similar reason: the overconfidence occurred in the bin [0.9, 1]. Generally, we observe a trade-off between conservatism and calibration, and given that weighted average is the least conservative composition, it predictably improves its performance from $\lambda=0.5$ to $\lambda=0.8$. %

\section{Discussion and Conclusion} \label{sec:discuss}

\edit{Our} theoretical and empirical investigation highlighted the distinct advantages of the proposed \edit{\coco} framework. First, unlike many assurance methods, it does not require a detailed or complete enumeration/model of hazards and failure modes. As long as the models and assumptions are safety-relevant, and the monitors are \edit{well-}calibrated and accurate, the composite monitor should detect any potential violation of safety. Second, it is not tied to specific closed-form distributions --- at most, our bounds only require the knowledge of the calibration bounds and monitor variance. Third, the non-data-driven compositions support independent development of monitors without the need for combined tuning. However, if joint \edit{monitor} data is available, compositions based on logistic regression and Bayesian estimation can take advantage of it. 

\looseness=-1
This paper is but an early step in compositional confidence monitoring, opening several exciting research directions. First, \edit{richer models of assumption dependencies (e.g., copulas~\cite{nelsen_introduction_2006}) may enable new composition functions (e.g., based on known odds or  correlations) and tighten the bounds on the existing ones}. Second, extending the composition bounds to 3+ monitors and new composition functions is \edit{mathematically} challenging, but it can be facilitated by detailed models of monitor dependencies. \edit{Another fruitful direction is empirical validation and refinement of the simplifications made in \Cref{sec:compconf}.} \edit{Since it is difficult to decompose assumptions, it would be useful to create verification techniques  that make granular assumptions bottom-up.}
Finally, a longer-term goal would be to extend the scope of a confidence monitor by considering its temporal behavior and monitoring the assumptions of other monitors~\cite{henzinger_monitorability_2020}.

To conclude, this paper presented an approach for run-time safety prediction by monitoring confidence in the assumptions of formal verification. %
The proposed \edit{\coco} framework introduced theoretical requirements for \edit{compositional} bounds on calibration errors, and instantiated these bounds for two composition functions. Two case studies demonstrated the \edit{practical} usefulness of \edit{\coco}. Furthermore, we observed that the compositions can adjust their conservatism, as well as improve their performance given synchronized data about multiple monitors and assumptions. %

\begin{acks}
\looseness=-1
\edit{The authors thank James Weimer, Sooyong Jang, and Northrop Grumman Corporation for helping develop the UUV simulator. This work was supported in part by  ARO W911NF-20-1-0080, AFRL and DARPA FA8750-18-C-0090, and ONR N00014-20-1-2744. Any opinions, findings and conclusions or recommendations expressed in this material are those of the authors and do not necessarily reflect the views of the Air Force Research Laboratory (AFRL), the Army Research Office (ARO), the Defense Advanced Research Projects Agency (DARPA), the Office of Naval Research (ONR) or the Department of Defense, or the United States Government.} 
\end{acks}

\bibliographystyle{ACM-Reference-Format}
\balance
\bibliography{refs.bib}

\appendix

\clearpage
\section*{Appendix}

This appendix contains the proofs of our theorems.

\section{Proof of Theorem~\ref{thm:safebound}}
\label{app:safebound}
\begin{proof}

Throughout this proof, we use the fact that {\psize$\statevec\not\in\statevecs(\sysmod) \implies \statevec\not\in\statevecs(\sysmod_\assn)$}, hence 
{\psize$\prob(\statevec\not\in\statevecs(\sysmod)) \leq \prob(\statevec\not\in\statevecs(\sysmod_\assn))$}.
{\psize
\begin{align*}
&\prob(\prop(\statevec) = \true \mid \statevec \not\in \statevecs(\sysmod_\assn)) \\
&\hspace{10pt}
=\prob(\prop(\statevec) = \true \land \statevec \in \statevecs(\sysmod) \mid \statevec \not\in \statevecs(\sysmod_\assn)) \\&\hspace{80pt}
+ \prob(\prop(\statevec) = \true \land \statevec\not\in\statevecs(\sysmod) \mid \statevec \not\in \statevecs(\sysmod_\assn))
\end{align*}
}

The first summand is bounded by $e_2$ because $\assn$ is safety-relevant: %
{\psize
\begin{align*}
    &\prob(\prop(\statevec) = \true \land \statevec \in \statevecs(\sysmod) \mid \statevec \not\in \statevecs(\sysmod_\assn)) \\
    &\hspace{10pt}= \prob(\prop(\statevec) = \true \mid  \statevec \in \statevecs(\sysmod) \land \statevec \not\in \statevecs(\sysmod_\assn))~\prob( \statevec \in \statevecs(\sysmod) \mid \statevec \not\in \statevecs(\sysmod_\assn)) \\
    &\hspace{10pt}\leq \prob(\prop(\statevec) = \true \mid  \statevec \in \statevecs(\sysmod) \land \statevec \not\in \statevecs(\sysmod_\assn)) \leq e_2
\end{align*}
}

For the second summand, we apply Bayes' theorem:
{\psize
\begin{align*}
&\prob(\prop(\statevec) = \true \land \statevec\not\in\statevecs(\sysmod) \mid \statevec\not\in\statevecs(\sysmod_\assn)) \\
&\hspace{10pt}=\frac{
\prob(\statevec\not\in\statevecs(\sysmod_\assn) \mid \prop(\statevec) = \true \land \statevec\not\in\statevecs(\sysmod))~\prob(\prop(\statevec) = \true \land \statevec\not\in\statevecs(\sysmod))}{\prob(\statevec\not\in\statevecs(\sysmod_\assn))} \\
&\hspace{10pt}= \frac{\prob(\prop(\statevec) = \true \land \statevec\not\in\statevecs(\sysmod))}{\prob(\statevec\not\in\statevecs(\sysmod_\assn))}
\end{align*}
}

Recall that {\psize$\statevec \in \statevecs(\sysreal)$} by the precondition of the theorem and thus {\psize$\prob(\prop(\statevec) = \true \mid \statevec\not\in\statevecs(\sysmod)) = \prob(\prop(\statevec) = \true \mid \statevec\not\in\statevecs(\sysmod),\statevec \in \statevecs(\sysreal))$} Then, by the safety relevance of \sysmod, the above fraction is bounded:
{\psize
\begin{align*}
\frac{\prob(\prop(\statevec) = \true \mid \statevec\not\in\statevecs(\sysmod)) \prob(\statevec\not\in\statevecs(\sysmod))}{\prob(\statevec\not\in\statevecs(\sysmod_\assn))} \leq 
e_1 \frac{\prob(\statevec\not\in\statevecs(\sysmod))}{\prob(\statevec\not\in\statevecs(\sysmod_\assn))} \leq e_1
\end{align*}
}
\end{proof}

\section{Proof of Lemma~\ref{lem:prob-assump}}
\label{app:prob-assump}
\begin{proof}
\looseness=-1
We will build up the bounds for the expression under the expectation. 
Suppose monitors $\mon_1$ and $\mon_2$ have probability densities $\den_1(x)$ and $\den_2(y)$.
From \mcerr bounds, integration, and our conditional independence of assumptions and monitors, for $\mon_1$ we get:
{\psize
\begin{align*}
    x - e_1 &\leq  \prob(\assn_1 \mid \mon_1 = x )  \leq x + e_1, \\
     \smallint_0^1 (x - e_1)\den_1(x \mid \compmon )dx &\leq \smallint_0^1  \prob(\assn_1 \mid \mon_1 = x ) \den_1(x \mid \compmon )dx \\
     &\leq \smallint_0^1 (x + e_1)\den_1(x \mid \compmon )dx, 
     \\
     \e[\mon_1 \mid \compmon] -e_1 &\leq \prob(\assn_1 \mid \compmon)   \leq   \e[\mon_1 \mid \compmon]+e_1 \labelthis{eq:pra1-bound}
\end{align*}
}
Analogously, for $\mon_2$:
{\psize
\begin{align*}
    \e[\mon_2 \mid \compmon] -e_2 \leq \prob(\assn_2 \mid \compmon )  \leq   \e[\mon_2 \mid \compmon]+e_2 \labelthis{eq:pra2-bound}
\end{align*}
}
\end{proof}

\section{Proof of Theorem~\ref{th:ece-product}}
\label{app:ece-product}
\begin{proof}
From conditional independence: 
{\psize
\begin{align*}
    \ecerr(\compmon, \assn_1 \land \assn_2) &= \e[|  \prob(\assn_1 \land \assn_2 \mid \compmon ) - \compmon |]) \\
    &= \e[|  \prob(\assn_1 \mid \compmon) \prob( \assn_2 \mid \compmon ) - \compmon |]) 
\end{align*}
}

We split the proof into three cases: 
{
\begin{enumerate}
    \item Event $H_1$: {\psize$\e[\mon_1 \mid \compmon] \geq e_1$} and  {\psize$\e[\mon_2 \mid \compmon] \geq e_2$}
    \item Event $H_2$: {\psize$\e[\mon_1 \mid \compmon] < e_1$} and  {\psize$\e[\mon_2 \mid \compmon] < e_2$}
    \item Event $H_3$: {\psize$\e[\mon_1 \mid \compmon] < e_1$} xor  {\psize$\e[\mon_2 \mid \compmon] < e_2$}
\end{enumerate}
}

Then the expectation can be split accordingly:
{\psize
\begin{align*}
    &\e[|  \prob(\assn_1 \land \assn_2 \mid \compmon ) - \compmon |]) \labelthis{eq:h-case-split} \\
    &\hspace{10pt}=
    \sum_{i=1}^3\prob(H_i)\e[|  \prob(\assn_1 \land \assn_2 \mid \compmon ) - \compmon | \mid H_i]  
\end{align*}
}

To complete the proof, with the help of Lemma~\ref{lem:prob-assump}, we need to show that we can bound the conditional expectation by at least one term in the max under each case, i.e., for each $i$
{\psize
\begin{align*}
    &\e[|  \prob(\assn_1 \land \assn_2 \mid \compmon ) - \compmon | \mid H_i] \labelthis{eq:prod-bound} \\ 
    &\hspace{10pt}=
    \max [4e_1e_2, \sqrt{\var[\mon_1]\var[\mon_2]} + e_1 + e_2 + e_1e_2]
\end{align*}
}
For the sake of brevity, let {\psize$\cex{1} := \e[\mon_1 \mid \compmon]$} and {\psize$\cex{2} := \e[\mon_2 \mid \compmon]$}.

\textbf{Case $H_1$}:

Our $H_1$ restrictions allows multiplying inequalities \edit{\eqref{eq:pra1-bound} and \eqref{eq:pra2-bound}}
because all sides are non-negative: 
{\psize
\begin{align*}
     (\cex{1} -e_1)(\cex{2} -e_2) \leq \prob(\assn_1 \mid \compmon ) \prob(\assn_2 \mid \compmon ) \leq (\cex{1} +e_1)(\cex{2} +e_2),
\end{align*}
\begin{align*}
\intertext{ and then subtract \compmon:}
     (\cex{1} -e_1)(\cex{2} -e_2)  - \compmon &\leq \prob(\assn_1 \mid \compmon ) \prob(\assn_2 \mid \compmon ) - \compmon \\
     &\leq (\cex{1} +e_1)(\cex{2} +e_2) - \compmon
\end{align*}
}
     Taking the absolute value, using max inequality, and triangle inequality we get a bound on the expression under the expectation: 
{\psize
\begin{align*}
    \label{eq:pos-H-start}
     &|\prob(\assn_1 \mid \compmon ) \prob(\assn_2 \mid \compmon ) - \compmon |  \\ 
    &\hspace{10pt}\leq \max \Big[ | (\cex{1} -e_1)(\cex{2} - e_2) - \compmon|, |(\cex{1} +e_1)(\cex{2} +e_2) - \compmon|\Big] %
\end{align*}
}
    
Now we use two facts: {\psize $\compmon=\mon_1 \mon_2$} and {\psize $\e[\mon_1]\e[\mon_2]  = \psize \e[\mon_1\mon_2] ~- $} \\{\psize $\cov[\mon_1, \mon_2]$}. Then, we can proceed with the triangle inequality:
\edit{
{\psize 
\begin{align*}
     &\max \Big[ | (\cex{1} -e_1)(\cex{2} -e_2) - \compmon|, |(\cex{1} +e_1)(\cex{2} +e_2) - \compmon|\Big] 
     \\ 
     &\hspace{10pt}= \max \Big[ |\cov[\mon_1, \mon_2|\compmon]- e_1e_2 + e_1\cex{2} + e_2\cex{1} | , \\
     &\hspace{80pt} |\cov[\mon_1, \mon_2|\compmon] - e_1e_2 - e_1\cex{2} -  e_2\cex{1} | \big]  
     \\ 
     &\hspace{10pt}\leq |\cov[\mon_1, \mon_2|\compmon]| + e_1\cex{2} + e_2\cex{1} + e_1e_2 
\end{align*}}}
Using the Cauchy-\edit{Schwarz} inequality, our assumption on variances, and the fact that {\psize$\cex{1} \le 1, \cex{2} \le 1$} we get the final bound under Case $H_1$:  
{\psize
\begin{align*}
    &|\cov[\mon_1,  \mon_2|\compmon]| +  e_1\cex{2} + e_2\cex{1} + e_1e_2 
    \\
    & \hspace{60pt}
    \leq \sqrt{\var[\mon_1]\var[\mon_2]} + e_1 + e_2 + e_1e_2
\end{align*}
}

\textbf{Case $H_2$:} %
Recalling that {\psize $M_C = M_1M_2$}, note that
{\psize
\begin{align}
    M_C \leq \cex{1} < e_1,
\end{align}
}
since {\psize $M_1 \ge M_C$} everywhere. Also note that~\eqref{eq:pra1-bound} now becomes
{\psize
\begin{align}
    0 \le \prob(\assn_1 \mid \compmon)   \leq   2e_1.
\end{align}
}
Similarly,~\eqref{eq:pos-H-start} becomes
{\psize
\begin{align*}
|\prob(\assn_1 \mid \compmon ) \prob(\assn_2 \mid \compmon ) - \compmon | &\le \max[|-M_C|, |4e_1e_2 - M_C|] \\
&\le \max[e_1, e_2, 4e_1e_2]. 
\end{align*}
}
Note that both $e_1$ and $e_2$ are smaller than the bound under $H_1$, hence we only keep $4e_1e_2$ in the final bound.

\textbf{Case $H_3$}: 
Without loss of generality, consider the case when {\psize$\cex{1} < e_1$} and {\psize$\cex{2} \ge e_2$}. Then once again~\eqref{eq:pra1-bound} becomes
{\psize
\begin{align}
    0 \le \prob(\assn_1 \mid \compmon)   \leq   2e_1.
\end{align}
}
The bound in~\Cref{eq:pos-H-start} is now simplified \edit{only on one side}:
\edit{
{\psize
\begin{align*}
&|\prob(\assn_1 \mid \compmon ) \prob(\assn_2 \mid \compmon ) - \compmon | \\
&\hspace{60pt}\leq \max[|-M_C|, |(\cex{1} +e_1)(\cex{2} +e_2) - M_C|] \\ 
&\hspace{60pt}\leq \max [e_1, \sqrt{\var[\mon_1]\var[\mon_2]} + e_1 + e_2 + e_1e_2].
\end{align*}
}}
The case {\psize$\cex{1} \ge e_1$} and {\psize$\cex{2} < e_2$} is symmetric.
\end{proof}

\section{Proof of Theorem~\ref{th:ece-weighted}}
\label{app:ece-weighted}
\balance
\begin{proof}
Following the same structure as Theorem~\ref{th:ece-product}, the proof is split into three cases: 
{
\begin{enumerate}
    \item Event $H_1$: {\psize$\e[\mon_1 \mid \compmon] \geq e_1$} and  {\psize$\e[\mon_2 \mid \compmon] \geq e_2$}
    \item Event $H_2$: {\psize$\e[\mon_1 \mid \compmon] < e_1$} and  {\psize$\e[\mon_2 \mid \compmon] < e_2$}
    \item Event $H_3$: {\psize$\e[\mon_1 \mid \compmon] < e_1$} xor  {\psize$\e[\mon_2 \mid \compmon] < e_2$}
\end{enumerate}
}
Then we again split the expectation into \eqref{eq:h-case-split}.
To complete the proof, with the help of Lemma~\ref{lem:prob-assump}, we need to show that we can bound the conditional expectation by at least one term in the max under each case.
For the sake of brevity, let {\psize$\cex{1} := \e[\mon_1 \mid \compmon]$} and {\psize$\cex{2} := \e[\mon_2 \mid \compmon]$}.

\textbf{Case $H_1$ for averaging:}%
{\psize
\begin{align*}
&\max\big[| (\cex{1} -e_1)(\cex{2} -e_2) - \compmon|, |(\cex{1} +e_1)(\cex{2} +e_2) - \compmon| \big] \\
&\hspace{10pt}= \max\big[| (\cex{1} -e_1)(\cex{2} -e_2) - w_1\cex{1} - w_2\cex{2}|,  \\
&\hspace{80pt} |(\cex{1} +e_1)(\cex{2} +e_2) - w_1\cex{1} - w_2\cex{2}| \big]
\end{align*}
}

Now note that max[a, |b|] = max[a, b, -b], and also note that in our case, a > -b:
{\psize
\begin{align*}
 \max&\big[ w_1\cex{1} + w_2\cex{1} - (\cex{1} - e_1)(\cex{2} - e_2) , \\
 &\hspace{80pt}|(\cex{1} +e_1)(\cex{2} +e_2) - w_1\cex{1} - w_2\cex{2}| \big] \\
 &= \max\big[w_1\cex{1} + w_2\cex{2} - (\cex{1} - e_1)(\cex{2} - e_2) , \\
 &\hspace{80pt}(\cex{1} +e_1)(\cex{2} +e_2) - w_1\cex{1} - w_2\cex{2} \big] \\
 &= \max\big[w_1\cex{1} + w_2\cex{2} - \cex{1}\cex{2} + e_2\cex{1} + e_1\cex{2} - e_1e_2,  \\
 &\hspace{80pt}\cex{1}\cex{2} - w_1\cex{1} - w_2\cex{2} + e_2\cex{1} + e_1\cex{2} + e_1e_2 \big] 
\end{align*}
}

The final step is to recognize that {\psize$w_1 a + w_2 b - ab \in [0, \max[w_1,w_2]]$}. This can be revealed by solving an optimization problem. Thus, we can bound the first expression above with:
{\psize
\begin{align*}
\max[\max[w_1,w_2] + e_1 +e_2-e_1e_2, e_1 + e_2 + e_1e_2]
\end{align*}
}

\textbf{Case $H_2$ for averaging:} 
Here, we are bounded by 
{\psize $$\max\big[| \compmon|, 
|(\cex{1} +e_1)(\cex{2} +e_2) - \compmon| \big],$$
}
which is itself bounded by \compmon since the values in the second argument are both positive. Hence, {\psize$\compmon= w_1\cex{1} + w_2\cex{2} \leq  w_1e_1 + w_2e_2 $}. This bound is dominated by  {\psize$e_1 + e_2 + e_1e_2$} from case $H_1$.

\textbf{Case $H_3$ for averaging:} Suppose {\psize$\cex{1} \leq e_1, \cex{2} \geq e_2$}. Then we need to bound: 
{\psize
\begin{align*}
    \max\big[| \compmon|, |(\cex{1} +e_1)(\cex{2} +e_2) - \compmon| \big] %
    = \compmon %
    = w_1\cex{1} + w_2\cex{2} \leq w_1e_1 + e_2
\end{align*}
}
Analogously, in the alternative case the bound is {\psize$e_1 + w_2e_2$}. Both of these bounds are lower than {\psize$e_1 + e_2 + e_1e_2$} from case $H_1$. 

Similar to the product bound proof, the ultimate bound is: 
{\psize
\begin{align*}
    \max[e_1 + e_2 + e_1e_2, \max[w_1,w_2] + e_1 +e_2-e_1e_2 ]
\end{align*}
}
\end{proof}

\section{Proof of Theorem~\ref{th:cce-product}}
\label{app:cce-product}
\begin{proof}
When {\psize$\compmon = x = \mon_1 \mon_2$}, it is clear that {\psize$\mon_1 \geq x, \mon_2 \geq x$} because monitors take values between 0 and 1. It follows that: 
{\psize
\begin{align*}
    \e[\mon_1 \mid \mon_1 \mon_2 = x] \geq x, \hspace{30pt} \e[\mon_2 \mid \mon_1 \mon_2 = x] \geq x
\end{align*}
}

Following Lemma~\ref{lem:prob-assump}, we can see that: 
{\psize
\begin{align*}
    \prob(\assn_1 \mid \compmon = x) \geq \e[\mon_1 \mid \mon_1 \mon_2 = x] - e_1 \geq x - e_1 \\ 
    \prob(\assn_2 \mid \compmon = x) \geq \e[\mon_1 \mid \mon_1 \mon_2 = x] - e_2 \geq x - e_2
\end{align*}
}

To ensure multiplicability: 
{\psize
\begin{align*}
     \prob(\assn_1 \mid \compmon = x) \geq \max[0, x - e_1] \\
    \prob(\assn_2 \mid \compmon = x) \geq \max[0, x - e_2] 
\end{align*}
}

Therefore:
{\psize
\begin{align*}
    \prob(\assn_1 \land \assn_2 \mid \compmon = x) \geq \max[0, x - e_1] \max[0, x - e_2]
\end{align*}
}
\end{proof}

\section{Proof of Theorem~\ref{th:ece-comp-bounds}}
\label{app:ece-comp-bounds}
\begin{proof}

{\psize
\begin{align*}
&\ecerr(\compmon, \safe) = \e[|\prob(\safe \mid \compmon) - \compmon|] \\
&\hspace{5pt}=\e[|\prob(\safe \mid \compmon) - \compmon|] - \ecerr(\compmon,\compassn) + \ecerr(\compmon,\compassn) \\
&\hspace{5pt}=\e[|\prob(\safe \mid \compmon) - \compmon|] - \e[|\prob(\compassn \mid \compmon) - \compmon|] + \ecerr(\compmon,\compassn) \\
&\hspace{5pt}\leq \e[|\prob(\safe \mid \compmon) - \prob(\compassn \mid \compmon)|] + \ecerr(\compmon,\compassn)
\end{align*}
}
The second summand is bounded by $e_3$ by the condition of the theorem. Consider now the first summand:
{\psize
\begin{align*}
    &\e[|\prob(\safe \mid \compmon) - \prob(\compassn \mid \compmon)|] \\
    &\hspace{40pt}= \e[|\prob(\safe \land \compassn \mid \compmon) + \prob(\safe \land \neg\compassn \mid \compmon) \\
    &\hspace{80pt}- \prob(\compassn \land \safe \mid \compmon) - \prob(\compassn \land \neg\safe \mid \compmon)|] \\
    &\hspace{40pt}= \e[|\prob(\safe \land \neg\compassn \mid \compmon)  - \prob(\compassn \land \neg\safe \mid \compmon)|]
\end{align*}
}
\edit{Note that {\psize$\prob(\compassn \land \neg\safe \mid \compmon) = 0$} because $\compassn$ is}
sufficient for proving safety and thus $\compassn$ and $\neg\safe$ are mutually exclusive, leaving us with: %
{\psize
\begin{align*}
\e[|\prob(\safe \land \neg\compassn \mid \compmon)|] = \prob(\safe \land \neg\compassn) 
&=\prob(\safe \mid \neg\compassn) \prob(\neg\compassn) \\
&\leq \prob(\safe \mid \neg\compassn) \leq e_1 + e_2,
\end{align*}
}
\edit{where the last bound $e_1 + e_2$ is a result of Theorem~\ref{thm:safebound}.}
\end{proof}

\section{Proof of Theorem~\ref{th:cce-comp-bounds}}
\label{app:cce-comp-bounds}

\edit{
\begin{proof}
{\psize
\begin{align*}
    &\max_{x \in [0,1]} [ x- \prob(\safe \mid \compmon = x)]\\
&\hspace{10pt}=\max_{x \in [0,1]} [ x - \prob(\assn \mid \compmon = x) \prob(\safe \mid \assn, \compmon = x) \\
&\hspace{80pt}-\prob(\neg \assn \mid \compmon = x) \prob(\safe \mid \neg \assn, \compmon = x)]\\
&\hspace{10pt}\leq \max_{x \in [0,1]} [ x - \prob(\assn \mid \compmon = x)]\leq e 
\end{align*}}
We are justified in using the upper bound {\psize $\max_{x \in [0,1]} [ x - \prob(\assn \mid \compmon = x)]$} as a conservative approximation because (i) \assn is sufficient, so {\psize$\prob(\safe \mid \assn, \compmon = x) = 1$}, and (ii) the last summand {\psize $\prob(\neg \assn \mid \compmon = x) \cdot \prob(\safe \mid \neg \assn, \compmon = x)$} is non-negative. 
\end{proof}
}

\end{document}